\documentclass[12pt, onecolumn,a4paper,unpublished]{quantumarticle}
\pdfoutput=1 
\usepackage[T1]{fontenc}
\usepackage{makecell}
\usepackage{tikz}
\usetikzlibrary{positioning, calc}
\usetikzlibrary{arrows}
\usepackage{tikz-3dplot}
\usepackage{graphicx}
\usepackage{mdwlist}
\usepackage{color}
\usepackage{thmtools}
\usepackage{amssymb}
\usepackage{physics}
\usepackage{amsmath}
\usepackage{array}
\usepackage{doi}
\usepackage{url,hyperref}
\usepackage[capitalize,nameinlink]{cleveref}
\hypersetup{colorlinks={true},linkcolor={blue},citecolor=red}

\usepackage[numbers,sort&compress]{natbib}
\usetikzlibrary{fadings}
\usetikzlibrary{decorations.pathmorphing}
\usepackage{mathrsfs}
\usepackage{authblk}
\usepackage[mathscr]{euscript}
\usepackage{enumitem}

\DeclareMathAlphabet\mathbfcal{OMS}{cmsy}{b}{n}

\usetikzlibrary{decorations.pathreplacing}
\tikzset{snake it/.style={decorate, decoration=snake}}

\usetikzlibrary{decorations.pathreplacing,decorations.markings}

\tikzset{
    >=stealth',
    punkt/.style={
           rectangle,
           rounded corners,
           draw=black, very thick,
           text width=6.5em,
           minimum height=2em,
           text centered},
    pil/.style={
           ->,
           thick,
           shorten <=2pt,
           shorten >=2pt,},
  on each segment/.style={
    decorate,
    decoration={
      show path construction,
      moveto code={},
      lineto code={
        \path [#1]
        (\tikzinputsegmentfirst) -- (\tikzinputsegmentlast);
      },
      curveto code={
        \path [#1] (\tikzinputsegmentfirst)
        .. controls
        (\tikzinputsegmentsupporta) and (\tikzinputsegmentsupportb)
        ..
        (\tikzinputsegmentlast);
      },
      closepath code={
        \path [#1]
        (\tikzinputsegmentfirst) -- (\tikzinputsegmentlast);
      },
    },
  },
  mid arrow/.style={postaction={decorate,decoration={
        markings,
        mark=at position .5 with {\arrow[#1]{stealth'}}
      }}}
}

\usepackage{slashed}
\usepackage{float} 
\usepackage{subcaption}

\newtheorem{theorem}{Theorem}

\newtheorem{corollary}[theorem]{Corollary}

\newtheorem{definition}[theorem]{Definition}

\newtheorem{lemma}[theorem]{Lemma}

\newtheorem{remark}[theorem]{Remark}

\newenvironment{proof}[1][Proof]{\noindent\textbf{#1. }}{\ \rule{0.5em}{0.5em}}
\newcommand{\CNOT}{\text{CNOT}}
\newcommand{\SWAP}{\text{SWAP}}

\begin{document}

\author[1]{Richard Cleve}
\email{cleve@uwaterloo.ca}
\orcid{}

\author[2]{Alex May}
\email{amay@perimeterinstitute.ca}
\orcid{0000-0002-4030-5410}

\title{Lower bounds on non-local computation from controllable correlation}

\affiliation[1]{Institute for Quantum Computing and School of Computer Science, University of Waterloo}
\affiliation[2]{Institute for Quantum Computing and Perimeter Institute for Theoretical Physics}

\abstract{Understanding entanglement cost in non-local quantum computation (NLQC) is relevant to complexity, cryptography, gravity, and other areas. 
This entanglement cost is largely uncharacterized; previous lower bound techniques apply to narrowly defined cases, and proving lower bounds on most simple unitaries has remained open.
Here, we give two new lower bound techniques that can be evaluated for any unitary, based on their \emph{controllable correlation} and \emph{controllable entanglement}.
For Haar random two qubit unitaries, our techniques typically lead to non-trivial lower bounds. 
Further, we obtain lower bounds on most of the commonly studied two qubit quantum gates, including CNOT, DCNOT, $\sqrt{\SWAP}$, and the XX interaction, none of which previously had known lower bounds. 
For the CNOT gate, one of our techniques gives a tight lower bound, fully resolving its entanglement cost.
The resulting lower bounds have parallel repetition properties, and apply in the noisy setting.
}

\maketitle

\tableofcontents

\section{Introduction}

Suppose that Alice and Bob begin with system $A$ and $B$ respectively, which they would like to interact under unitary $U_{AB}$. 
To do this, one possibility is for Alice and Bob to carry their systems and meet somewhere, then directly interact their systems. 
Alternatively, they can execute a \emph{non-local quantum computation} (NLQC) protocol, as shown in figure \ref{fig:non-localandlocal}. 
In that setting, Alice and Bob share entanglement, act locally in their own labs, exchange a single simultaneous round of communication, then act locally again. 
In this case, $A$ and $B$ are never brought together, but entanglement and communication allow the simulation of a local interaction between them. 

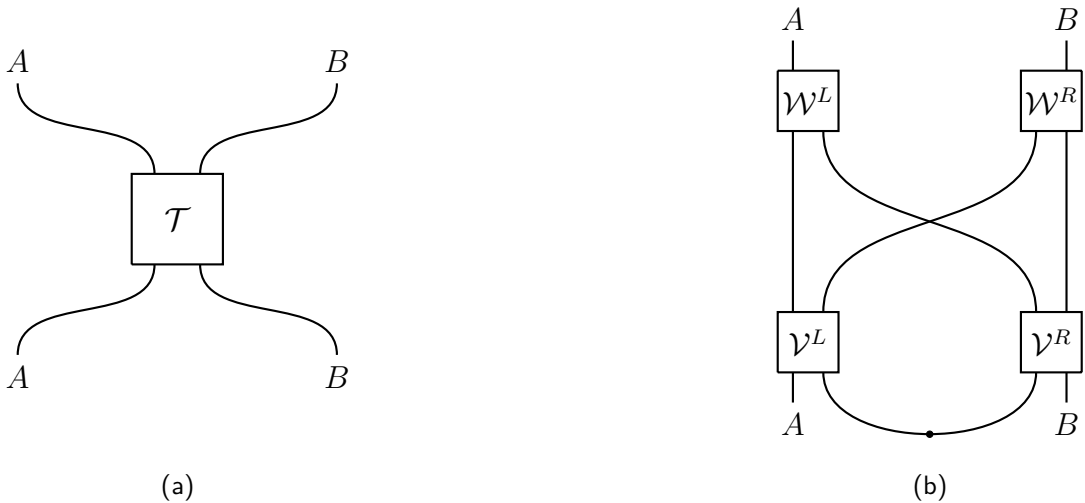
\begin{figure*}[ht]
    \centering
    \begin{subfigure}{0.45\textwidth}
    \centering
    \begin{tikzpicture}[scale=0.6]
    
    \draw[thick] (-1,-1) -- (-1,1) -- (1,1) -- (1,-1) -- (-1,-1);
    
    \draw[thick, mid arrow] (-3.5,-3) to [out=90,in=-90] (-0.5,-1);
    \node[below] at (-3.5,-3) {$A$};
    \draw[thick, mid arrow] (3.5,-3) to [out=90,in=-90] (0.5,-1);
    \node[below] at (3.5,-3) {$B$};
    
    \draw[thick, mid arrow] (0.5,1) to [out=90,in=-90] (3.5,3);
    \node[above] at (3.5,3) {$B$};
    \draw[thick, mid arrow] (-0.5,1) to [out=90,in=-90] (-3.5,3);
    \node[above] at (-3.5,3) {$A$};
    
    \node at (0,0) {$\mathcal{T}$};

    \node at (0,-5) {$ $};
    
    \end{tikzpicture}
    \caption{}
    \label{fig:local}
    \end{subfigure}
    \hfill
    \begin{subfigure}{0.45\textwidth}
    \centering
    \begin{tikzpicture}[scale=0.4]
    
    \draw[thick] (-5,-5) -- (-5,-3) -- (-3,-3) -- (-3,-5) -- (-5,-5);
    \node at (-4,-4) {$\mathcal{V}^L$};
    
    \draw[thick] (5,-5) -- (5,-3) -- (3,-3) -- (3,-5) -- (5,-5);
    \node at (4,-4) {$\mathcal{V}^R$};
    
    \draw[thick] (5,5) -- (5,3) -- (3,3) -- (3,5) -- (5,5);
    \node at (4,4) {$\mathcal{W}^R$};
    
    \draw[thick] (-5,5) -- (-5,3) -- (-3,3) -- (-3,5) -- (-5,5);
    \node at (-4,4) {$\mathcal{W}^L$};
    
    \draw[thick, mid arrow] (-4.5,-3) -- (-4.5,3);
    
    \draw[thick, mid arrow] (4.5,-3) -- (4.5,3);
    
    \draw[thick, mid arrow] (-3.5,-3) to [out=90,in=-90] (3.5,3);
    
    \draw[thick, mid arrow] (3.5,-3) to [out=90,in=-90] (-3.5,3);
    
    \draw[thick] (-3.5,-5) to [out=-90,in=-90] (3.5,-5);
    \draw[black] plot [mark=*, mark size=3] coordinates{(0,-7.05)};
    
    \draw[thick] (-4.5,-6) -- (-4.5,-5);
    \node[below] at (-4.5,-6) {$A$};
    \draw[thick] (4.5,-6) -- (4.5,-5);
    \node[below] at (4.5,-6) {$B$};
    
    \draw[thick] (4.5,5) -- (4.5,6);
    \node[above] at (-4.5,6) {$A$};
    \draw[thick] (-4.5,5) -- (-4.5,6);
    \node[above] at (4.5,6) {$B$};
    
    \end{tikzpicture}
    \caption{}
    \label{fig:non-localcomputation}
    \end{subfigure}
    \caption{Local and non-local computations. a) A unitary $U_{AB}$ is implemented by directly interacting the input systems. b) A non-local quantum computation. The goal is for the action of this circuit on the $AB$ systems to approximate $U_{AB}$.}
    \label{fig:non-localandlocal}
\end{figure*}

Non-local quantum computation has been related to quantum position-verification \cite{kent2011quantum,buhrman2014position}, quantum gravity \cite{may2019quantum,may2020holographic,may2023non,may2022complexity,may2022connected}, complexity theory \cite{buhrman2013garden,cree2022code,speelman2015instantaneous}, classical information-theoretic cryptography \cite{allerstorfer2024relating,asadi2024rank}, Hamiltonian complexity \cite{apel2024security}, uncloneable cryptography \cite{ananth2024unclonable}, and communication complexity \cite{buhrman2013garden,girish2025magic,asadi2025linear}. 
In all of these applications, understanding the entanglement cost of implementing a NLQC is the key relevant technical question. 
For instance, entanglement lower bounds are needed to establish security of quantum position-verification schemes and uncloneable secret sharing, lead to lower bounds on complexity, and lead to constraints on primitives studied in classical information-theoretic cryptography.
Despite the many applications of understanding entanglement cost in NLQC, doing so remains largely open. 

In this work, we give two new, related, lower bound techniques for NLQC. 
While previous techniques can be applied only to specific NLQC instances, the techniques here give lower bounds in terms of generic quantities that can be evaluated for any chosen unitary. 
This leads to lower bounds in many cases of interest. 

\subsection{Prior work}

NLQC was first studied in the context of quantum position-verification \cite{kent2011quantum,buhrman2014position}. 
In that context, NLQC upper bounds provide attacks on position-verification schemes, while lower bounds provide security guarantees. 
It was shown in \cite{buhrman2014position} that all unitaries can be implemented as an NLQC, although the protocol used $2^{2^{O(n)}}$ entanglement.
This was later reduced to $2^{O(n)}$ \cite{beigi2011simplified}. 
Another upper bound strategy is based on a Clifford+T decomposition of the unitary \cite{speelman2015instantaneous}. 

A popular class of NLQCs mix quantum and classical inputs. 
For instance, in $f$-routing \cite{kent2011quantum, buhrman2013garden}, Alice's input consists of a single qubit $Q$ along with a string $x\in\{0,1\}^n$, and Bob's input consists of a string $y\in\{0,1\}^n$. 
The players agree on a choice of Boolean function $f(x,y)$. 
The goal is for $Q$ to be held on Alice's side if $f(x,y)=0$, and brought to Bob's side if $f(x,y)=1$. 
For $f$-routing, upper bounds based on the complexity of the function $f$ have been developed \cite{buhrman2013garden, cree2022code}. 
In particular, the strongest such bound is based on the minimal size of a span program computing $f$ \cite{cree2022code}.  
It was later understood that $f$-routing is equivalent to other well studied NLQC protocols, so that these upper bounds can be extended to those settings \cite{bluhm2025complexity}.

Lower bounds on NLQC have proven difficult to develop, with the strongest (unconditional) lower bounds being linear in the input size. 
Further, even such linear lower bounds have only been proven in scattered cases. 
For a hidden basis task, a proof that entanglement is necessary was given in \cite{buhrman2014position}.
A parallel repetition for the same task was proven in \cite{tomamichel2013monogamy}; in particular they showed that repeating this task $n$ times leads to a success probability like $\beta^n$, $\beta<1$ when not using entanglement. 
This leads to a linear lower bound on entanglement cost via standard arguments, given explicitly in \cite{may2019quantum, may2020holographic}.
This bound continues to hold even in the approximate setting \cite{may2022connected}. 
For $f$-routing \cite{buhrman2013garden}, a lower bound on the dimension of the ancilla was given in \cite{bluhm2021position}, for random choices of function and in a purified model. 
In \cite{kon2025quantum}, this is improved on by, for instance, allowing general channels rather than unitaries in the attack model. 
Parallel repetition for $f$-routing and $f$-measure was established in \cite{escola2025quantum}.
In \cite{asadi2024rank} an entanglement lower bound linear in the classical input size was proven for some explicit, simple, choices of function $f$, although this relies on an assumption that the protocol is perfect in either $f(x,y)=0$ or $f(x,y)=1$ instances. 
Another approach instead considers the complexity of operations needed in the NLQC, which was lower bounded in \cite{asadi2025linear}. 

Regarding bounds that are tight when including numerical factors, \cite{ribeiro2015tight} gives a lower bound of $n-O(\log (n))$ for a protocol with an upper bound of $n$, although their lower bound is only against attackers using classical communication.\footnote{We consider the stronger attack model allowing quantum communication, which gives a stronger security notion for QPV, and is the relevant model in several of the applications of NLQC to other areas.}
Also in the classical communication model \cite{allerstorfer2022role} discusses a lower bound on a Bell state distinguishing game, which we believe has a 1 EPR pair lower bound matching the upper bound using their techniques, although this is not made explicit. 

In the related setting where the communication between Alice and Bob is required to be classical, the ability of a gate to create entanglement is a lower bound on its entanglement cost \cite{gonzales2019bounds}.
Our bounds are similar in spirit to this idea, but apply in the case where quantum communication is allowed. 

For the $\CNOT$ gate specifically, it is straightforward to give an upper bound of $E_f\leq 1$ on the entanglement cost. 
This uses a standard technique for implementing Clifford unitaries using Bell basis teleportation, see e.g. \cite{chakraborty2015practical}.
In \cite{may2022complexity}, a matching lower bound of 1 EPR pair is given in a restricted setting where only stabilizer resource states and Clifford operations are allowed in the protocol. 

\subsection{Our results}


In this work we introduce two related techniques for lower bounding entanglement cost in NLQC. 
We briefly describe the results obtained using the two techniques below:
\begin{enumerate}
    \item \textbf{Controllable correlation:} This technique gives lower bounds on most common gates and bounds for random two qubit unitaries with high probability. In practice, we find that the lower bound is at most $n_A/2$ for $n_A$ the number of input qubits, and we can prove it is at most $n_A$. This lower bound always satisfies parallel repetition.
    \item \textbf{Controllable entanglement:} This technique applies in a more limited set of cases, but can be as large as $n_A$. For the CNOT gate, it gives a tight lower bound of $E_f\geq 1$ where $E_f$ is the entanglement of formation of the resource state. This lower bound satisfies parallel repetition in some cases, including for the CNOT gate.
\end{enumerate}
The lower bounds we obtain for concrete unitaries are shown in the following table.
\begin{figure}[ht]
\centering
\begin{tabular}{| p{3.59cm} | p{2.75cm} | p{2.75cm} | p{2.10cm} |}
\hline
\textbf{Gate} & \textbf{Lower bound from CE} & \textbf{Lower bound from CC} & \textbf{Ref.\ state for CC}\\ 
\hline
CNOT & 1 & 0.5 & $\rho_{cc}$ or $\Psi^+$\\  
\hline
DCNOT & 0 & 0.5 & $\rho_{cc}$ or $\Psi^+$\\
\hline
Berkeley B & 0.601 &  0.5 & $\rho_{cc}$ \\
\hline
$\exp(-i \frac{\pi}{4}X\otimes X)$ & 1 & 0.5 & $\rho_{cc}$ or $\Psi^+$ \\
\hline
i\SWAP & 0 & 0.5 & $\rho_{cc}$ or $\Psi^+$ \\
\hline
$\sqrt{\SWAP}$ & 0 & 0.30 & $\Psi^+$ \\
\hline
Sycamore & 0 & 0.48 & $\rho_{cc}$ or $\Psi^+$ \\
\hline
Magic & 0 & 0.5 & $\rho_{cc}$ or $\Psi^+$\\
\hline
Dagwood Bumstead & 0 & 0.08& $\Psi^+$ \\
\hline
CS & 0 & 0.30 & $\Psi^+$ \\
\hline
CT & 0 & 0.12 & $\Psi^+$ \\
\hline
ECR & 0 & 0.5 & $\Psi^+$ \\
\hline
CSX & 0 & 0.30 & $\Psi^+$ \\
\hline
Random unitary & 0 & $\langle (\lambda_1-\lambda_2)/2\rangle \approx 0.230$ & $\rho_{cc}$ or $\Psi^+$ \\
\hline
\end{tabular}
\caption{Results of a numerical optimization computing the controllable entanglement and controllable correlation for some simple gates; these values are lower bounds on the entanglement of formation in any resource state that suffices to complete the corresponding gate as an NLQC. The reference state is the choice of state on $QA$ used in deriving the lower bound; see figure \ref{fig:NLQC_correlation}. $\rho_{cc}$ refers to the maximally classically correlated pair of qubits, while $\Psi^+$ is a Bell state. Matrix expressions for the listed gates are in appendix \ref{sec:gates}. }\label{table:lowerbounds}
\end{figure}

Note that none of these gates had known lower bounds prior to this work. 
Our lower bounds can be applied broadly in that, given a unitary, there is a straightforward calculation one can perform and, when the calculation returns a positive value, we immediately have a lower bound on the entanglement of any resource state that suffices to implement the gate as a NLQC. 
This contrasts with previous techniques which applied to restrictive subsets of cases, and where each case is treated with separate arguments.

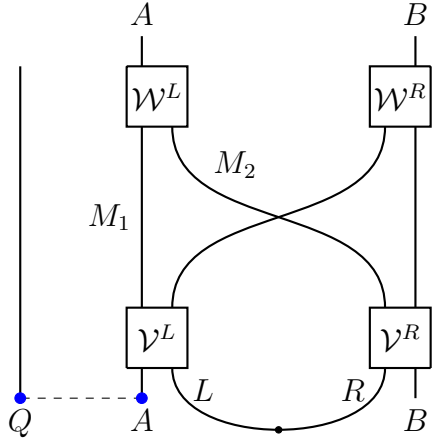
\begin{figure*}
    \centering
    \begin{tikzpicture}[scale=0.4]
    
    \draw[thick] (-5,-5) -- (-5,-3) -- (-3,-3) -- (-3,-5) -- (-5,-5);
    \node at (-4,-4) {$\mathcal{V}^L$};
    
    \draw[thick] (5,-5) -- (5,-3) -- (3,-3) -- (3,-5) -- (5,-5);
    \node at (4,-4) {$\mathcal{V}^R$};
    
    \draw[thick] (5,5) -- (5,3) -- (3,3) -- (3,5) -- (5,5);
    \node at (4,4) {$\mathcal{W}^R$};
    
    \draw[thick] (-5,5) -- (-5,3) -- (-3,3) -- (-3,5) -- (-5,5);
    \node at (-4,4) {$\mathcal{W}^L$};
    
    \draw[thick] (-4.5,-3) -- (-4.5,3);
    
    \draw[thick] (4.5,-3) -- (4.5,3);
    
    \draw[thick] (-3.5,-3) to [out=90,in=-90] (3.5,3);
    
    \draw[thick] (3.5,-3) to [out=90,in=-90] (-3.5,3);
    
    \draw[thick] (-3.5,-5) to [out=-90,in=-90] (3.5,-5);
    \draw[black] plot [mark=*, mark size=3] coordinates{(0,-7.05)};
    \node[below right] at (-3.2,-5) {$L$};
    \node[below left] at (3.2,-5) {$R$};
    
    \draw[thick] (-4.5,-6) -- (-4.5,-5);
    \node[below] at (-4.5,-6) {$A$};
    \draw[thick] (4.5,-6) -- (4.5,-5);
    \node[below] at (4.5,-6) {$B$};
    
    \draw[thick] (4.5,5) -- (4.5,6);
    \draw[thick] (-4.5,5) -- (-4.5,6);
    \node[above] at (-4.5,6) {$A$};
    \node[above] at (4.5,6) {$B$};

    \node[left] at (-4.5,0) {$M_1$};
    \node[right] at (-2.5,1.75) {$M_2$};

    \draw[thick] (-8.5,-6) -- (-8.5,5);
    \draw[dashed] (-4.5,-6) -- (-8.5,-6);
    \draw[blue] plot [mark=*, mark size=5] coordinates{(-4.5,-6)};
    \draw[blue] plot [mark=*, mark size=5] coordinates{(-8.5,-6)};
    \node[below] at (-8.5,-6) {$Q$};
    
    \end{tikzpicture}
    \caption{A non-local quantum computation implementing a unitary $U_{AB}$. To prove lower bounds on the entanglement cost, we consider placing the input system $A$ in a state $P_{QA}$ correlated with a reference system $Q$. We indicate this with the dashed line. The state $P_{QA}$ need not be pure. We find that if adjusting the input on $B$ changes the amount of correlation between $A$ and $Q$ in the final state, that $L:R$ must be entangled.}
    \label{fig:NLQC_correlation}
\end{figure*}

The starting point for the lower bounds is the set-up shown in figure \ref{fig:NLQC_correlation}. 
We consider the two input wires to the quantum operation of interest. 
Consider a correlated state $P_{QA}$. 
We refer to $Q$ as the reference system, $A$ as the input, and $B$ as the control.  
Then, we are interested in the total correlation or entanglement between $Q$ and $A$ after we run the circuit. 
If the total correlation or entanglement in this state can be varied by adjusting the input to the control wire $B$, then we obtain a lower bound.

In more detail, in the controllable correlation technique we are interested in a unitary $U_{AB}$ and how much we can vary the total correlation in $Q\!:\!A$ by adjusting $B$. 
Take any choice of initial correlated state $P_{QA}$.
Then, defining
\begin{align}
    \lambda_1 &= \max_{\phi^1_B} I(Q:A)_{U_{AB}(P_{QA}\otimes \phi^1_B)U_{AB}^\dagger}, \nonumber \\
    \lambda_2 &= \min_{\phi^2_B} I(Q:A)_{U_{AB}(P_{QA}\otimes \phi^2_B)U_{AB}^\dagger},
\end{align}
we find that
\begin{align}
    \boxed{E_f(L:R)_{\Psi} \geq \frac{\lambda_1-\lambda_2}{2}}
\end{align}
where $E_f(L:R)_{\Psi}$ is the entanglement of formation in the distributed resource system $\Psi_{LR}$ used in the NLQC protocol. 

In the controllable entanglement technique we are again interested in a unitary $U_{AB}$. 
We define
\begin{align}
    \lambda := \max_{\phi_B} E_f(Q:A)_{U_{AB}(\Psi^+_{QA}\otimes \phi_B)U_{AB}^\dagger}\,,\quad \lambda' = \min_{\phi_B} E_f(Q:A)_{U_{AB}(\Psi^+_{QA}\otimes \phi_B)U_{AB}^\dagger}.
\end{align}
Then we say that $U_{AB}$ is $\lambda$-controllably correlated if $\lambda>0, \lambda'=0$. 
For any $\lambda$ controllably correlated unitary, we obtain the lower bound
\begin{align}
    \boxed{E_f(L:R)_\Psi \geq \lambda}.
\end{align} 
Both this and the bound from the controllable correlation can be adapted to bound noisy implementations of the relevant unitary. 
We give the full bounds with error terms in the main text.

Both lower bound techniques have convenient parallel repetition properties. 
For the controllable correlation, if $G$ has lower bound $\Delta\lambda\equiv (\lambda_1-\lambda_2)/2$ then we find $G^{\otimes n}$ has lower bound at least $n\Delta \lambda$. 
For the controllable entanglement, a lower bound of $\lambda$ on a unitary $G$ implies a lower bound of $n\lambda$ on $G^{\otimes n}$. 

An important special case is the CNOT gate. 
For CNOT, the controllable entanglement technique has $\lambda=1$ $\lambda'=0$. 
Thus we obtain a tight lower bound of $E_f(L:R)_\Psi\geq n$ to implement $n$ instances of the CNOT gate in parallel. 

A limitation of our results is that we only bound the entanglement cost of unitary operations; new techniques would be needed to address general quantum channels. 
We comment further on this in the discussion.

\section{Quantum information tools}

We recall some quantum information theory tools that we make use of, and fix our notation. 

\subsection{Quantum states and distances}

We label the dimension of a Hilbert space $\mathcal{H}_A$ by $d_A$, and the (base 2) log dimension by $n_A=\log d_A$.
Throughout this work $\log$ denotes the base 2 logarithm, while $\ln$ denotes the natural logarithm.
When considering entanglement across bipartitions of a quantum state $\ket{\psi}_{AB}$, we refer to entanglement across $A:B$ where the colon indicates the partitioning of the systems. 
We use the notation 
\begin{align}
    \ket{\Psi^+}_{AB}=\frac{1}{\sqrt{2}}\left(\ket{00}_{AB}+\ket{11}_{AB} \right)
\end{align}
for this particular maximally entangled state of two qubits, and the notation
\begin{align}
    \rho_{cc} = \frac{1}{2}\left(\ketbra{00}{00}+\ketbra{11}{11} \right)
\end{align}
for the maximally classically correlated state of two qubits. 

We quantify the distance between quantum states with the one-norm distance, 
\begin{align}
    \Vert\rho-\sigma\Vert_1 = \tr|\rho-\sigma|.
\end{align}
Note that $\Vert\rho-\sigma\Vert_1/2$ is known as the trace distance. 

We quantify the distance between quantum channels using the diamond norm distance. 
\begin{definition} Let $\mathbfcal{N}_{B\rightarrow C}, \mathbfcal{M}_{B\rightarrow C}: \mathcal{L}(\mathcal{H}_B)\rightarrow \mathcal{L}(\mathcal{H}_C)$ be quantum channels. 
The \textbf{diamond norm distance} is defined by 
\begin{align}
    \Vert\mathbfcal{N}_{B\rightarrow C}-\mathbfcal{M}_{B\rightarrow C}\Vert_\diamond = \sup_{d} \max_{\Psi_{A_dB}}\Vert\mathbfcal{N}_{B\rightarrow C}(\Psi_{A_dB}) - \mathbfcal{M}_{B\rightarrow C}(\Psi_{A_dB})\Vert_1
\end{align}
where $\mathcal{H}_{A_d}$ is a $d$ dimensional Hilbert space. 
\end{definition}

We define the fidelity of quantum states by
\begin{align}
    F(\rho,\sigma)=\Vert \sqrt{\rho}\sqrt{\sigma}\Vert^2_1.
\end{align}
This is related to the trace distance by the Fuchs--Van de Graaf inequalities, 
\begin{align}\label{eq:FVdG}
    1-\sqrt{F(\rho,\sigma)}\leq \frac{1}{2}\Vert \rho-\sigma \Vert_1 \leq \sqrt{1-F(\rho,\sigma)}.
\end{align}
We next define an average case notion of the fidelity, which captures how well a quantum channel preserves information on average.
\begin{definition}
The average case fidelity of a quantum channel $\mathcal{N}_A$ is defined as
\begin{align}
    \bar{F}(\mathcal{N}_A) = \int d\psi \, F(\psi_A, \mathcal{N}_A(\psi))
\end{align}
where the integral is over the Haar measure. 
\end{definition}

Finally, we need the following statement. 
\begin{lemma}Given a quantum channel $\mathcal{N}$ acting on a $d$ dimensional space, we have
\begin{align}
    \frac{d+1}{d}\left(1-\bar{F}(\mathcal{N})\right)\leq \frac{1}{2}\Vert \mathcal{N}-\mathcal{I}\Vert_\diamond \leq \sqrt{d(d+1)}\sqrt{1-\bar{F}(\mathcal{N})}.
\end{align}
\end{lemma}
This is proposition 9 in \cite{wallman2014randomized}. 

With these notions and results in hand we can prove the following claim. 

\begin{lemma}\label{lemma:Uextreme} 
    Suppose that $\{\mathcal{N}^i\}_i$ are quantum channels, let $U$ be a unitary, and let $\mathcal{U}(\cdot)=U(\cdot)U^\dagger$ be a unitary channel, all acting on a $d$ dimensional Hilbert space. 
    Then if $\Vert \sum_i p_i \mathcal{N}^i-\mathcal{U}\Vert_\diamond \leq \epsilon$, then we have that $\sum_i p_i \Vert \mathcal{N}^i-\mathcal{U}\Vert_\diamond \leq d\sqrt{2\epsilon}$.
\end{lemma}
\begin{proof}
    We have by assumption that
    \begin{align}
        \epsilon \geq \left\Vert\sum_i p_i\mathcal{N}^i - \mathcal{U}\right\Vert_\diamond.
    \end{align}
    We would like to bound a similar quantity but with the sum moved outside the diamond norm. 
    To do this, we first relate the above to the average case fidelity, 
    \begin{align}\label{eq:epsilonFbar}
        \epsilon \geq \left\Vert\sum_i p_i\mathcal{N}^i - \mathcal{U}\right\Vert_\diamond = \left\Vert\sum_i p_i\mathcal{U}^\dagger \circ \mathcal{N}^i - \mathcal{I}\right\Vert_\diamond \geq 2\frac{d+1}{d}\left(1-\bar{F}\left(\sum_ip_i\mathcal{U}^\dagger\circ \mathcal{N}^i\right)\right).
    \end{align}
    Now we use that the average case fidelity is linear, in the sense that
    \begin{align}
        \bar{F}\left(\sum_ip_i\mathcal{U}^\dagger\circ \mathcal{N}^i\right) &= \int d\psi \, F\left( \sum_ip_i\mathcal{U}^\dagger\circ \mathcal{N}^i(\psi), \psi\right) \nonumber \\
        &=\sum_i p_i \int d\psi F\left( \mathcal{U}^\dagger\circ \mathcal{N}^i(\psi), \psi\right) \nonumber \\
        &= \sum_i p_i\bar{F}(\mathcal{U}^\dagger \circ \mathcal{N}^i)
    \end{align}
    The first equality used that the integral is over pure states. 
    Returning to equation \eqref{eq:epsilonFbar}, we have now
    \begin{align}
        \epsilon \geq 2\frac{d+1}{d} \left(1-\sum_ip_i\bar{F}\left(\mathcal{U}^\dagger\circ \mathcal{N}^i\right)\right).
    \end{align}
    or equivalently,
    \begin{align}
        \sum_ip_i\bar{F}\left(\mathcal{U}^\dagger\circ \mathcal{N}^i\right) \geq 1-\frac{1}{2}\frac{d}{d+1} \epsilon
    \end{align}
    Now we consider the quantity we want to upper bound, which is $\sum_i p_i\Vert \mathcal{U}^\dagger\circ \mathcal{N}^i-\mathcal{I}\Vert_\diamond$, 
    \begin{align}
        \sum_i p_i\Vert \mathcal{U}^\dagger\circ \mathcal{N}^i-\mathcal{I}\Vert_\diamond \leq 2\sqrt{d(d+1)}\sum_i p_i\sqrt{1-\bar{F}(\mathcal{U}^\dagger\circ \mathcal{N}^i)}
    \end{align}
    Now, use that $f(x)=\sqrt{1-x}$ is concave to move the sum inside the square root, to obtain
    \begin{align}
        \sum_i p_i\Vert \mathcal{U}^\dagger\circ \mathcal{N}^i-\mathcal{I}\Vert_\diamond \leq 2\sqrt{d(d+1)}\sqrt{1-\sum_i p_i\bar{F}(\mathcal{U}^\dagger\circ \mathcal{N}^i)} \leq d\sqrt{2\epsilon}
    \end{align}
    as claimed. 
\end{proof}

We will also make use of the following lemma, which is a simple consequence of the Fuchs--Van de Graaf inequalities along with Uhlmann's theorem. 

\begin{lemma}\label{lemma:traceUhlmann}
    Suppose that $\Vert \sigma_A-\rho_A\Vert_1 \leq \epsilon$, and consider any extension of $\sigma_A$ to the AB Hilbert space, call it $\sigma_{AB}$. Then there exists an extension of $\rho_A$ to the $AB$ Hilbert space, call it $\rho_{AB}$, such that
    \begin{align}
        \Vert \sigma_{AB}-\rho_{AB}\Vert_1\leq 2\sqrt{\epsilon}
    \end{align}
\end{lemma}
\begin{proof}
    Starting with $\Vert \sigma_A-\rho_A\Vert_1 \leq \epsilon$, use Fuchs--Van de Graaf to bound the fidelity from below, 
    \begin{align}
        F(\sigma_A,\rho_A) \geq 1-\epsilon.
    \end{align}
    Now consider any purification of $\sigma_{AB}$ into the $ABX$ Hilbert space, call it $\ket{\psi_\sigma}_{ABX}$. 
    Then by Uhlmann's theorem we have that there exists a state $\ket{\psi_\rho}_{ABX}$ such that
    \begin{align}
        F(\sigma_A,\rho_A)=|\braket{\psi_\sigma}{\psi_\rho}|^2.
    \end{align}
    But then we also have that the fidelity increases under the partial trace, so that
    \begin{align}
        F(\sigma_{AB},\rho_{AB})\geq |\braket{\psi_\sigma}{\psi_\rho}|^2 \geq 1-\epsilon.
    \end{align}
    Here $\rho_{AB}$ is defined by tracing out $X$ from $\ket{\psi_\rho}_{ABX}$. 
    Now we use Fuchs--Van de Graaf again to bound the trace distance between $\sigma_{AB}$ and $\rho_{AB}$, giving
    \begin{align}
        \Vert \sigma_{AB}-\rho_{AB}\Vert_1 \leq 2\sqrt{\epsilon}
    \end{align}
    as needed. 
\end{proof}

\subsection{Entropy and entanglement}\label{sec:entropyreview}

The von Neumann entropy of a state $\rho_A$ is 
\begin{align}
    S(A)_\rho = -\tr\rho_A\log \rho_A.
\end{align}
For two density matrices $\rho,\sigma$ with $\ker\sigma \subseteq \ker \rho$, the relative entropy is defined as
\begin{align}
    D(\rho||\sigma) = \tr(\rho\log \rho-\rho\log \sigma).
\end{align}
The relative entropy is related to the trace distance by Pinsker's inequality.
\begin{lemma}\textbf{(Quantum Pinsker inequality)}
The relative entropy $D(\rho||\sigma)$ and the one-norm $\Vert\rho-\sigma\Vert_1$ are related by
\begin{align}
    \frac{1}{2\ln 2}\Vert \rho-\sigma\Vert_1^2 \leq D(\rho||\sigma).
\end{align}
\end{lemma}
The mutual information is defined by
\begin{align}
    I(A:B)_\rho = S(A)_\rho+S(B)_\rho - S(AB)_\rho.
\end{align}
The mutual information satisfies the following continuity property \cite{winter2016tight}. 
\begin{lemma}\label{lemma:MIcontinuity} Suppose that $\Vert\rho-\sigma\Vert_1=\epsilon$. 
Then
\begin{align}
    |I(A:B)_\rho-I(A:B)_\sigma| \leq 4n_A\epsilon  + (1+2\epsilon)h\left( \frac{2\epsilon}{1+2\epsilon}\right).
\end{align}
\end{lemma}
We also make use of the conditional quantum mutual information, 
\begin{align}
    I(A:B|C)_\rho = S(AC)_\rho + S(BC)_\rho - S(C)_\rho - S(ABC)_\rho.
\end{align}
The mutual and conditional mutual informations are related by the chain rule, 
\begin{align}
    I(A:BC)_\rho = I(A:B|C)_\rho+ I(A:C)_\rho
\end{align}
We have the following statement about the conditional mutual information.
\begin{lemma}\label{lemma:dataprocessing} The quantum conditional mutual information satisfies the data processing inequality, 
    \begin{align}
    I(A:B|C)_\rho \geq I(A:B|C)_{\mathcal{N}_{B}(\rho)}.
\end{align}
\end{lemma}
\begin{proof}
This is more commonly stated for the mutual information (corresponding to $C=\varnothing$) but the statement for the conditional mutual information follows immediately, 
\begin{align}
    I(A:B|C)_\rho &= I(A:BC)_\rho - I(A:C)_\rho \nonumber \\
    &\geq I(A:BC)_{\mathcal{N}_B(\rho)} - I(A:C)_\rho \nonumber \\
    &=I(A:BC)_{\mathcal{N}_B(\rho)} - I(A:C)_{\mathcal{N}_B(\rho)} \nonumber \\
    &=I(A:B|C)_{\mathcal{N}_B(\rho)}.
\end{align}
We use the chain rule in the first and last lines, data processing for the mutual information in the second line, the third line is trivial, and the chain rule again in the last line.
\end{proof}

We next introduce the entanglement of formation \cite{hill1997entanglement,wootters1998entanglement} as a tool for quantifying entanglement. 
\begin{definition}
    The \textbf{entanglement of formation} is defined as
\begin{align}
    E_f(A:B)_\rho = \min_{\{p_i,\ket{\psi_i}\}} \sum_i p_i S(B)_{\psi_i},
\end{align}
where the minimization is over ensembles $\{p_i,\ket{\psi_i}\}$ such that $\rho=\sum_i p_i\ketbra{\psi_i}{\psi_i}$. 
\end{definition}
Note that an equivalent definition would replace $S(B)_{\psi_i}$ with $S(A)_{\psi_i}$. 
The entanglement of formation is a faithful measure of entanglement, meaning that it is zero if and only if $\rho$ is separable. 
This is easy to see from its definition: if it is separable so that 
\begin{align}
    \rho_{AB} = \sum_i p_i \rho_A^i\otimes \rho^i_B
\end{align}
then we introduce decompositions $\rho_A^i=\sum_k \lambda_k^i \ketbra{\phi_k^i}{\phi_k^i}_A$ and $\rho_B^i=\sum_k \mu_k^i \ketbra{\varphi_k^i}{\varphi_k^i}_B$ and we see that $E_f(A:B)_\rho=0$. 
Conversely, if $E_f(A:B)_\rho=0$ then there exists a decomposition into states $\psi_i$ such that $S(B)_{\psi_i}=0$ for all $i$, which means all $\psi_i$ are tensor product, and hence $\rho$ is separable. 

The entanglement of formation satisfies the following property, which shows that it can't grow too much as you add subsystems.
\begin{lemma} The entanglement of formation satisfies
    \begin{align}\label{eq:SRcontinuity}
    E_f(A:BC)_\rho\leq E_f(A:B)_\rho + \log d_C.
\end{align}
\end{lemma}
This statement follows from the definition of the entanglement of formation, subadditivity of the von Neumann entropy, and the statement $S(X)\leq \log d_X$. 

We also have a data processing inequality for the entanglement of formation \cite{bennett1996mixed}.
\begin{lemma}\label{lemma:Efdataprocessing} The entanglement of formation is non-increasing under the action of a local quantum channel,\footnote{In fact the entanglement of formation is also decreasing under LOCC operations, although we will not need that stronger property here.} 
\begin{align}
    E_f(A:B)_{\rho_{AB}} \geq E_f(A:B')_{\mathcal{I}_{A}\otimes \mathcal{N}_{B\rightarrow B'}(\rho_{AB})}.
\end{align}
\end{lemma}
Finally, we recall a Fannes type continuity bound \cite{nielsen2000continuity, winter2016tight} for the entanglement of formation. 
\begin{lemma}\label{lemma:Efcontinuity}
Consider two states $\rho_{AB}$, $\sigma_{AB}$ with $\Vert\sigma_{AB}-\rho_{AB}\Vert\leq \epsilon$, define $\eta_\epsilon=2\sqrt{\epsilon(1-\epsilon)}$ and $d=\min\{d_A,d_B\}$.
Then, the entanglement of formation of $\rho_{AB}$ and $\sigma_{AB}$ cannot be too different:
\begin{align}
    |E_f(A:B)_\rho - E_f(A:B)_\sigma| &\leq \eta_\epsilon \log d + H(\eta_\epsilon)
\end{align}
where $H(x) = (1+x)h\left(\frac{x}{1+x}\right)$ and $h(x)=-x\log x -(1-x)\log (1-x)$. 
\end{lemma}

We will also make brief use of a second entanglement measure, defined next. 
\begin{definition}\label{def:ER}
The relative entropy of entanglement is defined as
\begin{align}
    E_R(A:B)_\rho = \min_{\sigma_{AB}\in SEP_{AB}} D(\rho_{AB}||\sigma_{AB}),
\end{align}
where $D(\rho||\sigma)=\Tr(\rho\log\rho-\rho\log\sigma)$ is the quantum relative entropy and $SEP_{AB}$ is the set of separable states on $AB$.
\end{definition}

\section{Lower bound from controllable correlation}\label{sec:correlationlowerbound}

In this section we discuss the first of our two lower bounds, which considers the controllable correlation. 

\subsection{Proof of the lower bound}

In this section we give our first lower bound technique. 
The technique involves correlating one of the inputs to the gate of interest with a reference system in a state $P_{QA}$.
This state must be correlated across $Q\!:\!A$, but otherwise we leave our choice of this state free for now and fix it in concrete examples later.
We obtain a lower bound on entanglement in the resource state if it is possible to control whether correlation in $Q\!:\!A$ is preserved or destroyed by choosing the state of the input on $B$. 
We first give the following definition, which captures this notion of controlling correlation more precisely. 
\begin{definition}
    Consider a unitary $U_{AB}$ and choose a state on $QA$, which we label $P_{QA}$, with $n_Q=n_A$.
    Define the states
    \begin{align}
        \rho^1_{QAB}&=U_{AB}(P_{QA}\otimes \phi^1_B)U_{AB}^\dagger, \nonumber \\
        \rho^2_{QAB}&=U_{AB}(P_{QA}\otimes \phi^2_B )U_{AB}^\dagger.
    \end{align}
    We say that $U_{AB}$ has $(\lambda_1,\lambda_2)$-controllable correlation if there exists states $\phi^1_B, \phi^2_B$, and $P_{QA}$ such that
    \begin{align}
        \lambda_1=I(Q:A)_{\rho^1}, \qquad\lambda_2=I(Q:A)_{\rho^2}.
    \end{align}
    We always consider the case where $\lambda_1\geq \lambda_2$. If there is no choice of states $P_{QA}$, $\phi_B^1,\phi_B^2$ such that $\lambda_1>\lambda_2$ we say $U_{AB}$ is not controllably correlated. 
\end{definition}

As a simple example, the CNOT$_{B\rightarrow A}$ gate is controllably correlated: choose for instance $P_{QA} = (\rho_{cc})_{QA}=\frac{1}{2}(\ketbra{00}{00}_{QA}+\ketbra{11}{11}_{QA})$. Taking first the control on $B$ to be $\ket{0}$, CNOT$_{B\rightarrow A}$ acts identically on $A$, leaving $QA$ in the maximally classically correlated state $P_{QA}$, so $\lambda_1$ is $1$. 
On the other hand choosing the input on $B$ to be the maximally mixed state erases the state on $A$ and leaves it product with $Q$, so $\lambda_2=0$. 
In contrast, the SWAP$_{AB}$ gate is not controllably correlated --- regardless of the input on $B$, the final state on $QA$ will be product, so we always have $\lambda_1=\lambda_2=0$.
Similarly, the identity is not controllably correlated since $P_{QA}$ will not be influenced by the input on $B$.

Our main result of this section is the following theorem, which expresses that controllably correlated unitaries require entanglement to be implemented as a NLQC. 
System labels used in the proof are shown in figure \ref{fig:NLQC_correlation}.
\begin{theorem}\label{thm:controllablecorrelation}
    Suppose that unitary $U_{AB}$ has $(\lambda_1,\lambda_2)$-controllable correlation.
    If a NLQC protocol using a resource $\Psi$ gives an $\epsilon$-correct implementation of $U_{AB}$, then
    \begin{align}
        E_f(L:R)_\Psi \geq \frac{\lambda_1-\lambda_2}{2}-\Delta(2\sqrt{d_{AB}\sqrt{2\epsilon}},n_A)
    \end{align}
    where 
    \begin{align}
        \Delta(x,n_A) = 4n_A{x}+(1+{2x})h\left(\frac{{2x}}{1+{2x}} \right).
    \end{align}
\end{theorem}

\begin{proof}
    Recall that we defined the states
    \begin{align}
        \rho^1_{QAB}&=U_{AB}(P_{QA}\otimes \phi^1_B)U_{AB}^\dagger, \nonumber \\
        \rho^2_{QAB}&=U_{AB}(P_{QA}\otimes \phi^2_B )U_{AB}^\dagger.
    \end{align}
    These are the states resulting from the exact implementation of the unitary $U_{AB}$. 
    When $U_{AB}$ is replaced by the $\epsilon$-close implementation, we label the resulting states as $\sigma^1$ and $\sigma^2$, and note that we have
    \begin{align}
        \Vert\rho^1_{QAB}-\sigma^1_{QAB}\Vert_1\leq \epsilon, \nonumber \\
        \Vert\rho^2_{QAB}-\sigma^2_{QAB}\Vert_1\leq \epsilon,
    \end{align}
    which follows from the definition of the diamond norm distance.
    Note further that we write $\sigma_{QM_1M_2}^{1,2}$ for the states produced mid-way through the NLQC protocol (see figure \ref{fig:NLQC_correlation}) upon giving input $\phi^{1,2}_B$. 
    As well, we will drop the state label when considering the entropy of $Q$, since this is unaffected by the state on $B$. 
    Thus $S(Q)=S(Q)_{\sigma^1}=S(Q)_{\sigma^2}$. 

    We wish to understand how systems $M_1$ and $M_2$ are related to system $Q$. 
    First observe that by the causal structure of the circuit,
    \begin{align}
        \sigma_{M_2Q}^{i} = \sigma_{M_2}^i\otimes \rho_Q.
    \end{align}
    This holds regardless of the input on $B$, so for both $\sigma^1$ and $\sigma^2$.
    We will use this below in the form
    \begin{align}\label{eq:M2}
    \boxed{S(M_2Q)_{\sigma^1}=S(M_2)_{\sigma^1}+S(Q)_{\sigma^1}.}
    \end{align}
    
    Next, consider $M_1$.
    We have by assumption that
    \begin{align}
        I(Q:A)_{\rho^2}=\lambda_2.
    \end{align}
    We need to undo the last step of the NLQC circuit so as to construct the state on $M_1M_2$ from that on $A$.
     To do this, we consider taking a dilation of the channel $\mathcal{W}^L$ applied on Alice's side in the second round, and label the resulting unitary by $W^L_{M_1M_2\rightarrow AE}$ where $E$ is the ancillary system produced by the unitary. 
    This produces a density matrix $\sigma^2_{QAE}$.
    We claim this is close to $\rho^2_{QA}\otimes \kappa_E$ for some choice of density matrix $\kappa_E$. 
    To see why, recall that
    \begin{align}
        \Vert \sigma^2_{{Q}AB}-\rho^2_{{Q}AB} \Vert_1\leq \epsilon.
    \end{align}
    Considering the extension of $\sigma_{QAB}$ to $\sigma_{QABE}$, we apply lemma \ref{lemma:traceUhlmann} to find that there exists an extension of $\rho_{QAB}$ such that
    \begin{align}
        \Vert \sigma^2_{{Q}ABE}-\rho^2_{{Q}ABE} \Vert_1\leq 2\sqrt{\epsilon}.
    \end{align}
    But, then notice that $\rho^2_{QAB}$ is a pure state. 
    This means every extension must be of the form $\rho_{QAB}\otimes \kappa_E$ for some density matrix $\kappa_E$, so then
    \begin{align}
        \Vert\rho_{QA}^2\otimes \kappa_E - \sigma^2_{QAE}\Vert_1 \leq 2\sqrt{\epsilon}.
    \end{align}
    The state on $QM_1M_2$ then satisfies
    \begin{align}\label{eq:rooteps}
        \Vert(W^L_{M_1M_2\rightarrow AE})^\dagger( \rho_{{Q}A}^2\otimes \kappa_E)W^L_{M_1M_2\rightarrow AE} - \sigma^2_{{Q}M_1M_2}\Vert_1\leq 2\sqrt{\epsilon}.
    \end{align}
    Consider the state as above produced on giving input $\phi^2_B$, and consider the mutual information $I(Q:M_1)_{\sigma^2}$,
    \begin{align}
        I(Q:M_1)_{\sigma^2} &\leq I(Q:M_1M_2)_{\sigma^2} \nonumber \\ 
        &\leq I(Q:M_1M_2)_{\rho^2} + \Delta(2\sqrt{\epsilon},n_Q) \nonumber \\
        &=I(Q:A)_{\rho^2} + \Delta(2\sqrt{\epsilon},n_Q) \nonumber \\
        &=\lambda_2+\Delta(2\sqrt{\epsilon},n_Q).
    \end{align}
    The first inequality uses data processing, the second inequality uses equation \ref{eq:rooteps} and the continuity statement lemma \ref{lemma:MIcontinuity} and the third uses that $\rho_{QM_1M_2}=(W^L_{M_1M_2\rightarrow AE})^\dagger(\rho_{QA}^2\otimes \kappa_E)W^L_{M_1M_2\rightarrow AE}$.
    Finally, notice that the state on $M_1Q$ cannot depend on the input on $B$, so that
    \begin{align}
        I(Q:M_1)_{\sigma^1} = I(Q:M_1)_{\sigma^2} \leq \lambda_2 +\Delta(2\sqrt{\epsilon}, n_Q).
    \end{align}
    This statement is key to our proof and worth commenting on.
    This is telling us that even when inputting state $\phi^1_B$, the correlation across $Q\,:M_1$ must be small, and in particular similar to its value when inputting $\phi^2_B$.
    But, when the input state is $\phi^1_B$, a lot of correlation has to end up in $A$. 
    This means $M_1M_2$ is highly correlated with $Q$ even while $M_1$ is not. 
    We will use the above expression in the form
    \begin{align}\label{eq:M1}
        \boxed{S(M_1Q)_{\sigma^1}\geq S(M_1)_{\sigma^1}+S(Q)_{\sigma^1}-\lambda_2 - \Delta(2\sqrt{\epsilon},n_Q)}.
    \end{align}
    
    Continuing, we make use of the statement 
    \begin{align}
        I(Q:A)_{\rho^1}=\lambda_1.
    \end{align}
    To do so, we first use continuity of the mutual information to turn this into a statement about $\sigma^1$, 
    \begin{align}
        \lambda_1 - \Delta(\epsilon,n_Q)\leq I(Q:A)_{\sigma^1} .
    \end{align}
    Then observe that from data processing,
    \begin{align}
        I(Q:A)_{\sigma^1} \leq I(Q:M_1M_2)_{\sigma^1}
    \end{align}
    so then
    \begin{align}
        \lambda_1-\Delta(\epsilon,n_Q) \leq S(M_1M_2)_{\sigma^1}+S(Q)_{\sigma^1}-S(M_1M_2Q)_{\sigma^1}
    \end{align}
    or, rearranging,
    \begin{align}\label{eq:fixed}
        \boxed{S(M_1M_2Q)_{\sigma^1} \leq S(M_1M_2)_{\sigma^1}+S(Q) - \lambda_1 + \Delta(\epsilon,n_Q).}
    \end{align}
    We will use this below.

    Now, consider the conditional mutual information $I(M_1:M_2|Q)_{\sigma^1}$. 
    This is
    \begin{align}
        I(M_1:M_2|Q)_{\sigma^1} &= S(M_1Q)_{\sigma^1}+S(M_2Q)_{\sigma^1}-S(Q)_{\sigma^1}-S(M_1M_2Q)_{\sigma^1} \quad \qquad \qquad \,\,\,\,\text{definition of CMI} \nonumber \\
        &= S(M_1Q)_{\sigma^1}+S(M_2)_{\sigma^1}-S(M_1M_2Q)_{\sigma^1} \,\qquad \qquad\qquad\qquad \qquad \quad\,\,\,\,\,\,\,\,\,\,\, \qquad \text{eq. \eqref{eq:M2}} \nonumber \\
        &\geq S(M_1)_{\sigma^1} + S(Q) + S(M_2)_{\sigma^1} -\lambda_2 - S(M_1M_2Q)_{\sigma^1} -\Delta(2\sqrt{\epsilon},n_{Q}) \quad \quad \,\,\,\,\,\,\,\,\,\,\,\text{eq. \eqref{eq:M1}} \nonumber \\
        &\geq S(M_1)_{\sigma^1} + S(M_2)_{\sigma^1} - S(M_1M_2)_{\sigma^1} +\lambda_1 - \lambda_2 -\Delta(\epsilon,n_{Q})-\Delta(2\sqrt{\epsilon},n_{Q})\quad\,\text{eq. \eqref{eq:fixed}}\nonumber \\
        &\geq\lambda_1-\lambda_2 -2\Delta(2\sqrt{\epsilon},n_{Q})\,\,\,\,\quad  \qquad\qquad\qquad\qquad \,\,\qquad \qquad \quad\qquad\qquad \,\text{subadditivity} \nonumber 
    \end{align}
    so that the conditional mutual information is bounded below. In the last line we used that $\Delta(2\sqrt{\epsilon},n_{Q}) \geq \Delta(\epsilon,n_{Q})$, which holds because $\Delta$ is monotone increasing in the first argument, to simplify the error term. 

    Next, we would like to translate this to a bound on the mutual information of the resource state.
    Using that
    \begin{align}
        \sigma_{QM_1M_2}^1=\mathcal{V}^L_{AL\rightarrow M_1}\otimes \mathcal{V}^R_{RB\rightarrow M_2}(\sigma^{1}_{QALRB})
    \end{align}
    and data processing for the CMI (lemma \ref{lemma:dataprocessing}) we have that 
    \begin{align}
        I(M_1:M_2|Q)_{\sigma^1_{M_1M_2Q}}\leq I(RA:LB|Q)_{\sigma_{QALRB}^1} = I(R:L)
    \end{align}
    where in the second equality we used that
    \begin{align}
        \sigma^{1}_{QALRB} = P_{QA}\otimes \Psi_{LR}\otimes \phi^1_B.
    \end{align}
    Combined with our lower bound on the CMI, we have then
    \begin{align}
        \lambda_1-\lambda_2-2\Delta(2\sqrt{\epsilon},n_Q)\leq I(R:L)_{\Psi}.
    \end{align}
    Considering in particular a pure state resource, this leads to
    \begin{align}\label{eq:SRlowerboundCC}
        S(R)_\Psi \geq \frac{\lambda_1-\lambda_2}{2}-\Delta(2\sqrt{\epsilon},n_Q).
    \end{align}

    Finally, it remains to translate this to a lower bound on the entanglement of formation of the resource state. 
    For a mixed state resource, we can decompose the state as
    \begin{align}
        \Psi_{LR}=\sum_i p_i \ketbra{\Psi_i}{\Psi_i}_{LR}.
    \end{align}
    We will take the above to be the optimizing decomposition when computing the entanglement of formation for state $\Psi$. 
    Let $\mathcal{N}^i$ be the channel implemented by the NLQC protocol when given resource $\ket{\Psi_i}$. 
    Then the channel implemented given resource $\Psi$ is $\mathcal{N}=\sum_ip_i \mathcal{N}_i$, and we have by assumption that
    \begin{align}
        \epsilon \geq \Vert\mathcal{N} - U\Vert_\diamond = \left\Vert\sum_i p_i\mathcal{N}^i - U\right\Vert_\diamond.
    \end{align}
    From lemma \ref{lemma:Uextreme}, we have that
    \begin{align}
        \sum_i p_i\Vert \mathcal{N}^i-U\Vert_\diamond \leq d_{AB}\sqrt{2\epsilon}.
    \end{align}
    Now consider applying the lower bound on the entropy for each $\mathcal{N}^i$ separately. 
    Defining $\gamma_i=\Vert\mathcal{N}^i-U\Vert_\diamond$ we have that
    \begin{align}
        S(R)_{\Psi_i} \geq \frac{\lambda_1-\lambda_2}{2}-\Delta(2\sqrt{\gamma_i},n_Q).
    \end{align}
    Then the entanglement of formation is lower bounded by
    \begin{align}
        E_f(L:R)_\Psi &= \sum_i p_i S(R)_{\Psi_i} \nonumber \\
        &\geq \frac{\lambda_1-\lambda_2}{2} - \sum_i p_i \Delta(2\sqrt{\gamma_i},n_Q)\nonumber \\
        &\geq \frac{\lambda_1-\lambda_2}{2} -  \Delta\left(2\sqrt{\sum_ip_i\gamma_i},n_Q\right) \nonumber \\
        &\geq \frac{\lambda_1-\lambda_2}{2} -  \Delta(2\sqrt{d_{AB}\sqrt{2\epsilon}},n_Q)
    \end{align}
    as claimed. 
    In the second to last line we used that $\Delta(x,n)$ is concave in the first argument.
\end{proof}

Next, we comment on parallel repetition of lower bounds proven using the controllable correlation. 
\begin{corollary}
    Consider a unitary $G$ with $(\lambda_1,\lambda_2)$-controllable correlation.
    An exact implementation of $G^{\otimes n}$ as a NLQC requires entanglement of formation in the resource state $\Psi$ lower bounded by
    \begin{align}
        E_f(L:R)\geq n\left( \frac{\lambda_1-\lambda_2}{2}\right)
    \end{align}
    where $\lambda_1$, $\lambda_2$ are the parameters appearing in the controllable correlation for $G$. 
\end{corollary}
\begin{proof}
    Suppose that $\phi^1, \phi^2, P_{QA}$ can be used to show $G$ has $(\lambda_1, \lambda_2)$-controllable correlation. 
    Then considering $G^{\otimes n}$, use the correlated state $P_{QA}^{\otimes n}$, and inputs $(\phi^1)^{\otimes n}$, $(\phi^2)^{\otimes n}$, we obtain $\lambda_1'=n\lambda_1$, $\lambda_2'=n\lambda_2$, which lead to the stated lower bound.
\end{proof}

Finally, note that if we have a unitary $U_{AB}$ which is $\epsilon$ close to $G$ and implement it $n$ times, our techniques so far do not provide a good lower bound.
This is because the error in the implementation of $G^{\otimes n}$ becomes $n\gamma$, which eventually becomes larger than 1. 
Thus, we can so far only bound the parallel repeated setting in the case of zero error. 
However, if we assume that the noisy implementation of $U_{AB}$ is itself unitary, then we can simply use that it will have controllable correlation $(\lambda_1',\lambda_2')$ with $|\lambda_1-\lambda_1'|<\Delta(\epsilon, n_A)$, $|\lambda_2-\lambda_2'|<\Delta(\epsilon, n_A)$ which follows by continuity of the mutual information. 
This leads to a lower bound
\begin{align}
    E_f(L:R)\geq n\left( \frac{\lambda_1-\lambda_2}{2}-\Delta\right).
\end{align}
It is an open problem to obtain a similar lower bound in the case where we allow the noisy implementation to be a general quantum channel. 

\subsection{Evaluating the lower bound in simple cases}

For the case of two qubit gates, we provide code\footnote{\url{https://drive.google.com/file/d/18EGazQcjYsY2mlQVbEYVUoL3CJ_KbQDd/view?usp=sharing}} to compute the controllable correlation lower bound. 
We do this for a number of standard two qubit gates and report the results in table \ref{table:lowerbounds}.
Our code relies on a numerical optimization to select the states $\phi^1_B$ and $\phi^2_B$. 
This can be done quickly for the case of two qubit gates, but we note that because the function being optimized is non-convex (the mutual information), this may become difficult for larger Hilbert spaces. 

An interesting issue is the selection of the correlated state $P_{QA}$. 
We do not understand systematically which is the optimal choice of correlated state for a given gate.
In practice, we try both the maximally entangled state 
\begin{align}
    \ket{\Psi^+}_{QA} = \frac{1}{\sqrt{2}}(\ket{00}_{QA}+\ket{11}_{QA})
\end{align}
and the classically correlated state
\begin{align}
    \rho_{cc}=\frac{1}{2}(\ketbra{00}{00}_{QA}+\ketbra{11}{11}_{QA}).
\end{align}
We find that the lower bound using the two choices of state can be different, and even $0$ for one state while non-zero for the other. 

Figure \ref{fig:histogram} shows the results of computing numerical lower bounds, using the $\rho_{cc}$ state, for 100,000 samples drawn from the Haar distribution. 
A very similar distribution and average is obtained when using $\Psi^+_{QA}$.
We observe that the lower bound is never larger than $1/2$, and that we never find unitaries with zero entanglement cost\footnote{Or more precisely, we never find a lower bound that is within numerical precision of zero.}. 
Notice that since $\lambda_2\geq 0$ and $\lambda_1\leq 2n_B$, the lower bound provided by theorem \ref{thm:controllablecorrelation} is obviously never larger than $n_B$, so at most 1 for a two qubit gate. 
However, we do not have an analytical explanation for why the lower bound evaluates to at best $1/2$ in practice. 
The observation of not finding any unitaries with zero cost among these 100,000 samples suggests that the zero entanglement cost unitaries may be a set of measure zero. 
We leave investigating these observations to future work. 

\begin{figure}
    \centering
    \includegraphics[width=0.75\linewidth]{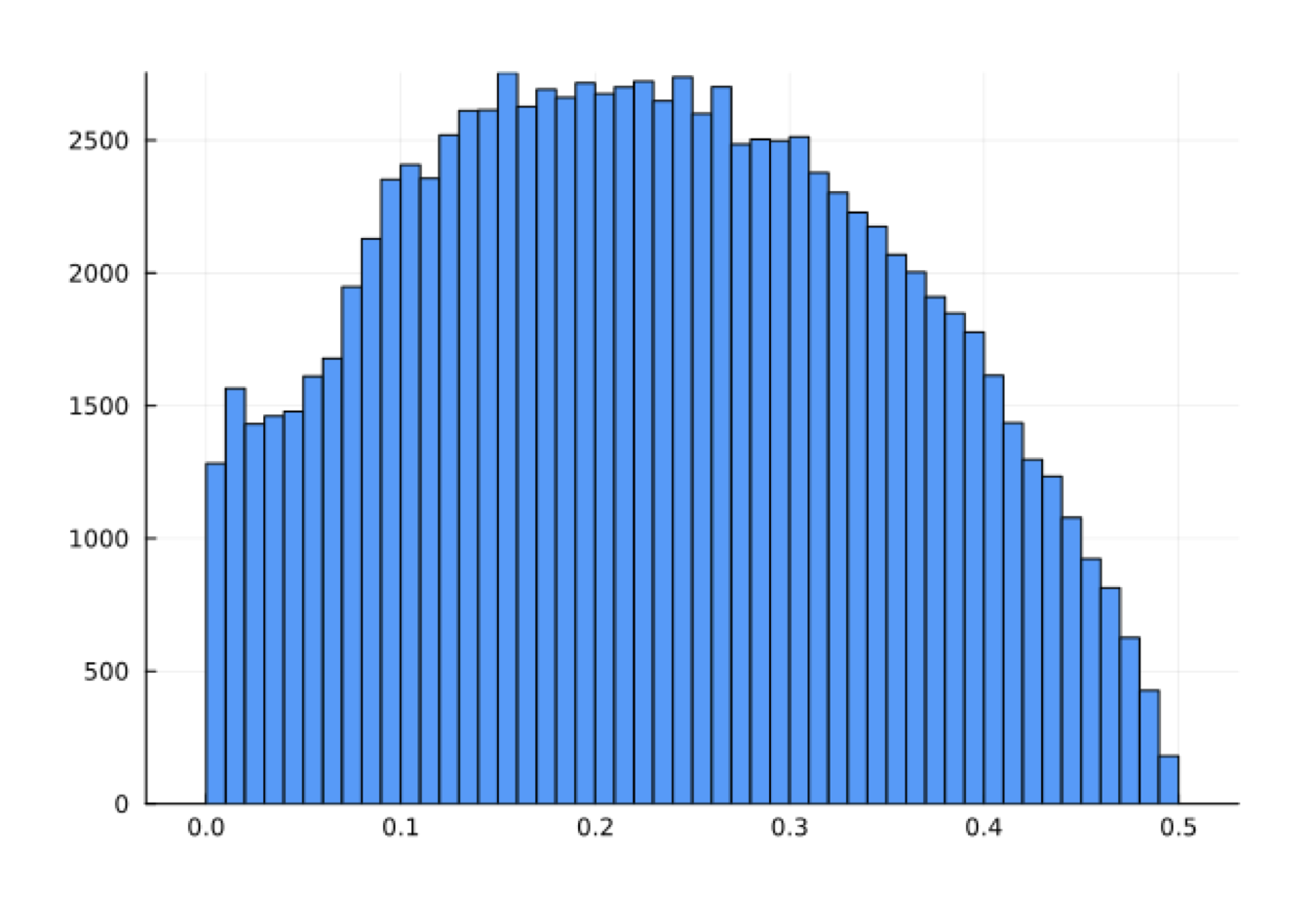}
    \caption{Histogram showing the value of the controllable correlation lower bound computed for $100,000$ two qubit unitaries drawn from the Haar distribution. The average value of the lower bound is $\approx 0.230$.}
    \label{fig:histogram}
\end{figure}

\section{Lower bound from the controllable entanglement}

In this section we discuss our second lower bound technique, the controllable entanglement. 
One motivation for considering this second technique is that the controllable correlation does not seem to produce tight lower bounds. 
For instance, it gives a lower bound of $1/2$ for the CNOT gate, but the best upper bound known is $1$.

\subsection{Proof overview}

To state the lower bound, we introduce the $\lambda$-\emph{controllable entanglement}.
The definition of the controllable entanglement makes use of the entanglement of formation, $E_f$, which is reviewed in section \ref{sec:entropyreview}.
\begin{definition} Consider a unitary $U_{AB}$. We say that $U_{AB}$ has $\lambda$-\textbf{controllable entanglement} if there exists states $\phi^1_B$, $\phi^2_B$ such that
\begin{align}
    \rho^1_{QAB}=U_{AB}(\Psi^+_{QA}\otimes \phi^1_B)U^\dagger_{AB},\qquad
    \rho^2_{QAB}=U_{AB}(\Psi^+_{QA}\otimes \phi^2_B)U^\dagger_{AB}.
\end{align}
with 
\begin{align}
    \lambda:=E_f(Q:A)_{\Psi_1}>0, \quad 0=E_f(Q:A)_{\rho^2}.
\end{align}
If there are no such choices of input state, we say that $U_{AB}$ is not controllably entangled. 
\end{definition}
In words, the controllable entanglement considers the following setting: We entangle one input, $A$, to the channel with a reference, $Q$. 
We consider varying the choices of input on $B$ to try and either make $Q\!:\!A$ very entangled or exactly separable. 

When $\lambda>0$, the controllable entanglement gives a lower bound on entanglement cost in NLQC. 
More precisely, we prove the following.
\begin{theorem}\label{thm:mainlowerbound}
Suppose that an NLQC protocol using resource state $\Psi_{LR}$ implements a unitary $U_{AB}$ to within diamond norm distance $\gamma$, and where $U_{AB}$ has $\lambda$-controllable entanglement.
Then for small enough $\gamma$, we have that
\begin{align}
    E_f(L:R)_\Psi \geq \lambda-P(d_{AB}\sqrt{2\gamma})
\end{align}
where $P(x)=2n_{Q}x^{1/8}$. 
\end{theorem} 
As an example application, consider the $\CNOT_{A\rightarrow B}$ gate. 
Taking $\ket{\phi^1}=\ket{+}$, we see that 
\begin{align}
    \ket{\Psi_1}=\CNOT_{AB}\ket{\Psi^+}_{QA}\ket{+}_B = \ket{\Psi^+}_{QA}\ket{+}_B
\end{align}
which leads to $\lambda=1$. 
Meanwhile, choosing $\ket{\phi^2}=\ket{0}$, we find
\begin{align}
    \ket{\Psi_2}=\CNOT_{AB}\ket{\Psi^+}_{QA}\ket{0}_B = \ket{GHZ}_{QAB}
\end{align}
This leaves $Q\!:\!A$ in a separable state, as needed. 
Thus the theorem above gives a lower bound of exactly 1 for the CNOT gate. 

In the remainder of this section we briefly outline the proof strategy for theorem \ref{thm:mainlowerbound}, and highlight how the subsections below correspond to steps in the proof. 
Refer to figure \ref{fig:NLQC_correlation} for the system labels we use. 

The first step of the proof is given in section \ref{sec:dimensionbound}. 
The reasoning begins with the functionality of the NLQC protocol, which implements the unitary $U_{AB}$ of interest. 
We consider the action of this unitary when $A$ is maximally entangled with a reference $Q$, and when we take two different choices of input to $B$, call them $\phi^1_B$ and $\phi^2_B$. 
We assume that the unitary is controllably entangled, so that the $Q\!:\!A$ system should is entangled when given $\phi^1_B$ and separable when given $\phi^2_B$. 
This gives statements about entanglement in the outputs, which we then need to translate to statements about entanglement in the resource. 
To do so, we use the data processing inequality and the causal structure of the circuit to understand the entanglement pattern in the mid-protocol state (the state after the first set of operations in figure \ref{fig:NLQC_correlation}, but before the second). 
From this, we find that the $Q:M_1$ wires must be unentangled, but $Q:M_1M_2$ must be strongly entangled when given input $\phi^1_B$. 
This then means that $R$ must have been large enough to carry the missing entanglement, and we obtain a lower bound on the \emph{dimension} of $R$, and in particular that $\log d_R\gtrsim \lambda$.

In section \ref{sec:purestatelowerbound}, we then continue by showing that in fact the entropy of $R$ must also be large. 
To do this, we introduce the parallel repetition of the protocol implementing $U_{AB}$. 
Repeating this $m$ times, our bound on dimension says that $R$ needs to consist of at least $m\lambda$ qubits. 
But if the entropy of $R$ is smaller than $\lambda$, we can compress it using Schumacher's compression scheme and take the compressed state as a resource for the implementing $U^{\otimes m}$. 
To avoid violating our dimension lower bound then, the resource system must not be compressible to fewer than $\sim m\lambda$ qubits, so we find that the \emph{entropy} of $R$ must be at least $\lambda$. 

The entropy lower bound is a good entanglement lower bound when the resource system is pure. 
To address the most general, mixed state, case, we upgrade this to a bound on the entanglement of formation. 
This is done in section \ref{sec:entanglementbound}. 
The key observation is that if the NLQC protocol works well on average, the entropy lower bound must apply to every state in any ensemble decomposition of the mixed state resource, and so every state in the ensemble decomposition must have a large entropy. 
But the entanglement of formation is defined in exactly this way: it is the minimal average entropy of pure states in any ensemble decomposition of $\Psi$. 
Since the entropy is bounded below by $\lambda$ for every term in the decomposition, we are led to the same lower bound for the entanglement of formation.  

\subsection{Proof of the lower bound}

\subsubsection{Dimension lower bound}\label{sec:dimensionbound}

As our first step towards a proof of theorem \ref{thm:mainlowerbound}, we prove the following lemma which lower bounds the dimension of the resource system.  
\begin{lemma}\label{lemma:dimensionlowerbound}
    Suppose $U_{AB}$ has $\lambda$-controllable entanglement. 
    A non-local quantum computation which implements $U_{AB}^{\otimes m}$ to within diamond norm $\epsilon$ and uses a resource system $\Psi_{LR}$ must have 
    \begin{align}
        n_R \geq m(\lambda-g_1(\epsilon)) -  g_2(\epsilon)
    \end{align}
    where $g_1(\epsilon),g_2(\epsilon)$ are positive functions that go to $0$ as $\epsilon\rightarrow 0$. More specifically, they are given by
    \begin{align}
        g_1(\epsilon) &= n_Q(\eta_{\epsilon} + \eta_{2\sqrt{\epsilon}}), \nonumber \\
        g_2(\epsilon) &= H(\eta_\epsilon) + H(\eta_{2\sqrt{\epsilon}}),
    \end{align}
    where $\eta_x=2\sqrt{x(1-x)}$ and $H(y)=(1+y)h\left(\frac{y}{1+y}\right)$.
\end{lemma}

\begin{proof} Suppose we have an NLQC protocol that implements $\tilde{U}$, with $\Vert\tilde{U}-U^{\otimes m}\Vert_\diamond\leq \epsilon$ using a resource state $\Psi_{LR}$. 
We take as input on the left one end of the maximally entangled state $\ket{\Psi^+}^{\otimes m}_{QA}$, and label the reference system by $Q$. 
The set up is shown in figure \ref{fig:NLQC_correlation}; we will use the operation and system labels shown there. 

We consider two scenarios. 
First, consider inputting the $\ket{\phi^1}^{\otimes m}$ state into the remaining input, labelled $B$. 
We label the state resulting from a perfect implementation of $U^{\otimes m}$ on this input as $\rho^1$, and from the imperfect implementation as $\sigma^1$. 
In this case, a perfect implementation of $U^{\otimes m}$ would lead to, by the definition of the controllable entanglement, 
\begin{align}
    E_f(Q:A)_{\rho^1} = \lambda m.
\end{align}
Since the protocol instead implements $\tilde{U}$ which is close to $U^{\otimes m}$, we need to use the continuity bound of lemma \ref{lemma:Efcontinuity} and we obtain
\begin{align}
    E_f(Q:A)_{\rho^1} - E_f(Q:A)_{\sigma^1} \leq m \,n_Q\,\eta_\epsilon + H(\eta_\epsilon).
\end{align}
Using $E_f(Q:A)_{\rho^1} = \lambda m$, we have then that the $A$ output from $U_{AB}$ is close to being $\lambda m$ entangled with $R$, 
\begin{align}
    E_f(Q:A)_{\sigma^1} \geq m(\lambda-n_Q\eta_\epsilon)-H(\eta_\epsilon).
\end{align}
We can also observe that, 
\begin{align}\label{eq:correct}
    E_f(Q:M_1M_2)_{\sigma^1}\geq m(\lambda-n_Q\eta_\epsilon)-H(\eta_\epsilon)
\end{align}
which follows from the previous line and the data processing inequality (lemma \ref{lemma:Efdataprocessing}) for the entanglement of formation. 

Second, we input $\ket{\phi^2}^{\otimes m}$ into $B$. 
We label the state created in this case by $\rho^2$ in the perfect case, and as $\sigma^2$ in the imperfect case.
Recall that $\rho^2_{QA}$ is separable.
Consider purifying the operation $\mathcal{W}^L_{M_1M_2\rightarrow A}$ to a unitary $W^L_{M_1M_2\rightarrow AE}$.
This produces a state $\sigma_{QAE}$. 
By lemma \ref{lemma:traceUhlmann}, we can extend the $\sigma^2,\rho^2$ states to the $QABEX$ Hilbert space ($X$ is an additional purifying system) and obtain
\begin{align}
    \Vert \sigma^2_{QABEX} - \rho^2_{QABEX} \Vert_1 \leq 2\sqrt{\epsilon}.
\end{align}
Since $\rho^2_{QAB}$ is pure, the extension of $\rho^2$ must be product across $QAB:EX$, so
\begin{align}
    \Vert \sigma^2_{QABEX} - \rho^2_{QAB}\otimes \rho_{EX} \Vert_1 \leq 2\sqrt{\epsilon}.
\end{align}
Next trace out $BX$, 
\begin{align}
    \Vert \sigma_{QAE}^2 - \rho_{QA}^2\otimes \rho_{E} \Vert_1 \leq 2\sqrt{\epsilon}.
\end{align}
Now apply $(W^L)^\dagger$ to both states, which won't change the trace distance, and choose an explicit decomposition of $\rho^2$ into a convex sum over product states (recall that by assumption it is separable), 
\begin{align}
    2\sqrt{\epsilon} &\geq \left\Vert\sum_{i}p^i\rho^i_{Q}\otimes (W^L)^\dagger_{M_1M_2\rightarrow AE}\rho^i_{A}\otimes \rho_E W^L_{M_1M_2\rightarrow AE} - (W^L)^\dagger_{M_1M_2\rightarrow AE}\sigma^2_{QA}\otimes \rho_EW^L_{M_1M_2\rightarrow AE}\right\Vert_1 \nonumber \\
    &=\left\Vert\sum_{i}p^i\rho^i_{Q}\otimes \rho^i_{M_1M_2} - \sigma^2_{QM_1M_2}\right\Vert_1
\end{align}
From this we also obtain that the state on $\sigma^2_{QM_1}$ is close to separable, which using continuity of $E_f$ gives, 
\begin{align}
    E_f(Q:M_1)_{\sigma^2} \leq m n_{Q}\eta_{2\sqrt{\epsilon}} + H(\eta_{2\sqrt{\epsilon}}).
\end{align}
Finally, notice that by causality the state on $QM_1$ must be the same regardless of the input at $B$, so that
\begin{align}\label{eq:smallonQM_1}
    E_f(Q:M_1)_{\sigma^1} \leq m\,n_{Q}\eta_{2\sqrt{\epsilon}} + H(\eta_{2\sqrt{\epsilon}}).
\end{align}
In words, we see that in the state $\sigma^1$ systems $Q:M_1$ are close to separable, while $Q:M_1M_2$ is entangled.

Now we combine our statements so far to show that this can only occur when $R$ is large enough,
\begin{align}
    m(\lambda-n_Q\eta_\epsilon)-H(\eta_\epsilon) &\leq E_f(Q:M_1M_2)_{\sigma^1} \qquad \qquad \,\,\,\,\,\,\,\,\,\text{From eq. \ref{eq:correct}}\nonumber \\
       &\leq E_f(Q:M_1RB')_{\sigma^1} \qquad \qquad \,\,\,\text{From data processing}\nonumber \\
       &= E_f(Q:M_1R)_{\sigma^1} \nonumber \qquad\qquad \,\,\,\,\,\,\,\,\,\,\text{Because $B'$ is tensor product} \\
       &\leq E_f(Q:M_1)_{\sigma^1} + n_R \nonumber \quad\,\,\,\,\,\,\qquad \text{From eq. \ref{eq:SRcontinuity}}\\
       &=m\,n_Q\eta_{2\sqrt{\epsilon}} + H(\eta_{2\sqrt{\epsilon}}) + n_R, \,\quad \,\,\,\, \text{From eq. \ref{eq:smallonQM_1}.}
\end{align}
so we have that 
\begin{align}
    n_R  \geq m(\lambda_1-n_Q\eta_\epsilon - n_Q\eta_{{2\sqrt{\epsilon}}} ) -H(\eta_\epsilon) - H(\eta_{2\sqrt{\epsilon}})
\end{align}
as claimed. 
\end{proof}

\subsubsection{Entropy lower bound for any pure state resource}\label{sec:purestatelowerbound}

In section \ref{sec:dimensionbound} we gave a lower bound on the number of qubits of resource system needed in an NLQC implementing a unitary $U_{AB}$ with the controllable entanglement property. 
In this section we translate this into a bound on the entanglement in the resource system, under the assumption that the resource system is pure, which we quantify using the entropy of one side of the resource state. 
We treat this first in the case where the implementation of $U$ is exact, then in the case where the implementation of $U$ is approximate. 
Our treatment of the approximate setting contains as a special case the exact one so in principle the exact case could be omitted, but the exact case is significantly simpler than the approximate one and conveys the key elements of the proof, so we retain it. 

\vspace{0.2cm}
\noindent \textbf{Exact case:} To upgrade our dimension lower bound to a lower bound on entropy, our approach is to make use of Schumacher compression \cite{schumacher1995quantum}, stated in the next theorem.

\begin{theorem}\label{thm:schumacher} \textbf{(Schumacher compression)}
    Suppose we have a quantum source which produces $\ket{\psi}_{LR}^{\otimes m}$. 
    Then, for all $\epsilon, \delta\in(0,1)$, there is a large enough $m$ such that there exists a compression map $\mathcal{C}_{R^m\rightarrow M}$ and decompression map $\mathcal{D}_{M\rightarrow R^m}$ with
    \begin{align}
        \Vert\ketbra{\psi}{\psi}^{\otimes m} - \mathcal{D}_{M\rightarrow R^m}\circ \mathcal{C}_{R^m\rightarrow M}(\ketbra{\psi}{\psi}^{\otimes m}) \Vert_1 \leq \epsilon
    \end{align}
    and where $\log d_M\leq (S(R)+\delta)m$.
\end{theorem}

We use this along with our lower bound on dimension, lemma \ref{lemma:dimensionlowerbound}, to obtain a lower bound on entropy. 
The basic idea is that the resource system $\ket{\psi}_{LR}$ can be compressed to contain $S(R)m$ qubits using Schumacher compression, but from our dimension bound we know the number of qubits must be $\lambda_1 m$, so we must have $S(R)\geq \lambda_1$. 
We give a more careful proof next. 

\begin{lemma}\label{lemma:exactcaseentropylowerbound}
    Suppose that an NLQC protocol implements $U_{AB}$ exactly, using a pure resource state $\ket{\psi}_{LR}$, and where $U_{AB}$ has $\lambda$-controllable entanglement. Then, $S(R)\geq\lambda$.
\end{lemma}

\begin{proof}
Using Schumacher compression (theorem \ref{thm:schumacher}), for any choice of $\epsilon, \delta>0$ there is an $m$ large enough, choice of compression channel $\mathcal{E}_{R^m\rightarrow M}$, and decompression channel $\mathcal{D}_{M\rightarrow R^m}$ with $n_M\leq (S(R)+\delta)m$, such that if we define
\begin{align}
    \Psi_{L^mM}=\mathcal{E}_{R^m\rightarrow M}(\ketbra{\psi}{\psi}^{\otimes m})
\end{align}
then
\begin{align}
    \left\Vert \ketbra{\psi}{\psi}^{\otimes m}_{LR} - \mathcal{D}_{M\rightarrow R^m}(\Psi_{L^mM})\right\Vert_1 \leq \epsilon.
\end{align}
We define an NLQC protocol to implement $U^{\otimes m}$ as follows. 
The distributed resource state is taken to be $\Psi_{L^mM}$. 
In the first set of operations, on the right, Bob applies $\mathcal{D}_{M\rightarrow R^m}$, leaving Alice and Bob sharing a state $\epsilon$-close to $\ket{\psi}^{\otimes m}$. 
Next, they run $m$ copies of the protocol, using the $m$ copies of $\ket{\psi}$ as resource states. 
By the properties of the diamond norm distance, this will be $\epsilon$-close in diamond norm to an implementation of $U^{\otimes m}$. 

Now we make use of lemma \ref{lemma:dimensionlowerbound}, which tells us that
\begin{align}
    n_M \geq m(\lambda-g_1(\epsilon)) - g_2(\epsilon)
\end{align}
But also, at large enough $m$, $(S(R)+\delta)m \geq n_M$, so that
\begin{align}\label{eq:SRlowerbound}
    S(R) \geq \lambda-\delta - g_1(\epsilon)-g_2(\epsilon)/m
\end{align}
But we can choose $\epsilon$, $\delta$ arbitrarily small while $m$ becomes arbitrarily large, so that we obtain $S(R)\geq \lambda$, as claimed.
\end{proof}

\vspace{0.2cm}
\noindent \textbf{Approximate case:} In the case where the NLQC protocol implements $U$ approximately, the asymptotic statement of Schumacher compression doesn't suffice to obtain a lower bound. 
The reason for this can be seen by considering equation \ref{eq:SRlowerbound}. 
There, the error $\epsilon$ in the implementation of $U^{\otimes m}$ comes from the approximation to the resource state appearing when decompressing from Schumacher's scheme. 
If each $U$ implementation is approximate, there is a contribution to the error from each $U$, so we would replace $\epsilon\rightarrow \epsilon + \gamma\,m$ where $\gamma$ is the error in a single $U$ implementation.
But $g(x)$ is only defined for $x\in[0,1]$, so as $m\rightarrow \infty$ we never have a lower bound and we do not obtain a bound in the asymptotic setting considered by Schumacher. 
To remedy this, we will need to consider Schumacher compression for a finite number of copies of the input state. 
This is addressed in \cite{abdelhadi2020second}; we briefly recall one of their results here. 

A compression protocol consists of a compression channel $\mathcal{C}_{A^n\rightarrow M}$ and a decompression channel $\mathcal{D}_{M\rightarrow A^n}$. 
We say the protocol is $\epsilon$-correct if the entanglement fidelity of the input $\rho^{\otimes n}$ is $\epsilon$-close to the entanglement fidelity of the output,
\begin{align}
    F_e(\rho_A^{\otimes n},\mathcal{D}\circ\mathcal{C}(\rho_A^{\otimes n})) \geq 1-\epsilon.
\end{align}
We denote the minimal log-dimension of $M$ needed to achieve $\epsilon$-correct compression on $n$ copies of $\rho_A$ by $M(n,\epsilon, \rho)$. 

The value of $M(n,\epsilon, \rho)$ is well understood. 
To state the result, we define the function on density matrices
\begin{align}
    V(A)_\rho = \tr(\rho_A\log^2\rho_A) - \left(\tr(\rho_A \log \rho_A)\right)^2
\end{align}
and the function
\begin{align}
    \Phi^{-1}(x) = \sup \left\{z\in \mathbb{R}:\frac{1}{\sqrt{2\pi}}\int_{-\infty}^{z} e^{-t^2/2} dt\leq x \right\}.
\end{align}
This is known as the \emph{quantile} of the normal distribution; it expresses how far we need to integrate the normal distribution with variance 1 to reach a given value $x$. 
The quantile of the normal distribution is defined on $(0,1)$ and diverges as $x\rightarrow 0,1$. 

Finally, we can state the following theorem, proven in \cite{abdelhadi2020second} as theorem 3. 
\begin{theorem}\label{thm:finite_n_schumacher}\textbf{(Schumacher compression at finite block length)}
The minimal achievable value of $M(n,\epsilon, \rho)$ in performing $\epsilon$-correct compression of the state $\rho_A$ is
\begin{align}
    M(n,\epsilon,\rho)= n S(A)_\rho + \Phi^{-1}(\sqrt{1-\epsilon})\sqrt{n V(A)_\rho} + O(\log n).
\end{align}
\end{theorem}
We use this along with lemma \ref{lemma:dimensionlowerbound} to obtain a lower bound on the entropy. 

\begin{theorem}\label{thm:entropylowerboundapproximate}
    Let $U_{AB}$ be $\lambda$-controllably entangled. 
    Suppose that an NLQC protocol implements $U_{AB}$ to within diamond norm distance $\gamma$, using a pure resource state $\ket{\psi}_{LR}$. Then, for small enough $\gamma$,
    \begin{align}
        S(R)_\psi\geq \lambda- P(\gamma)
    \end{align}
    where $P(\gamma)=2\sqrt{2}n_Q\gamma^{1/8}$.
\end{theorem}
\begin{proof}
We consider an implementation of $U^{\otimes m}$, where we will choose $m$ later. 
Our implementation uses as a resource state a compressed version of  $\ket{\psi}_{LR}^{\otimes m}$, that is we use 
\begin{align}
    \Psi_{L^mM} = \mathcal{I}\otimes \mathcal{E}_{R^m\rightarrow M}(\ketbra{\psi}{\psi}^{\otimes m})
\end{align}
where $\mathcal{E}_{R^m\rightarrow M}$ is an optimal compression channel. 
The protocol proceeds by first having Bob decompress $M$ into $R^m$, and then running $m$ parallel implementations of $U$ as before. 
We use an $\epsilon$-correct compression protocol where we choose $\epsilon$ later. 
Since the compression protocol is $\epsilon$-correct, and each $U$ implementation is $\gamma$-correct, by the properties of the diamond norm the implementation of $U^{\otimes m}$ using the compressed resource state will be $\epsilon+\gamma\,m$ correct. 

Now we make use of lemma \ref{lemma:dimensionlowerbound}, which tells us that
\begin{align}
    n_M \geq m(\lambda-g_1(\epsilon+\gamma m)) - g_2(\epsilon+\gamma m).
\end{align}
Now use theorem \ref{thm:finite_n_schumacher} as an upper bound on $n_M$, and using that $V(R)_\rho \leq n_R^2$, we have
\begin{align}
    S(R)_\psi \geq \lambda - \Phi^{-1}(\sqrt{1-\epsilon})\frac{n_R}{\sqrt{m}} - g_1(\epsilon + \gamma m) - \frac{g_2(\epsilon+\gamma m)}{m} - O\left(\frac{\log m}{m}\right).
\end{align}
For intuition, notice that if we take $\gamma=0$ we can maximize the lower bound by sending $\epsilon\rightarrow0,m\rightarrow \infty$, in which case we recover $S(R)_\rho\geq \lambda$. 
At non-zero $\gamma$ however, sending $m\rightarrow \infty$ would remove the lower bound, which only applies when the total error $\epsilon+ \gamma m\in[0,1]$. 
To recover a good lower bound, we need to choose $\epsilon, m$ in a way that depends on $\gamma$ such that the lower bound approaches $\lambda$ as $\gamma\rightarrow 0$.
We will achieve this with a simple choice by taking
\begin{align}
    \epsilon &= \gamma, \qquad m =\frac{1}{\sqrt{\gamma}}.
\end{align}
Inserting this above leads to the lower bound
\begin{align}
    S(R)_\psi \geq \lambda - \Phi^{-1}(\sqrt{1-\gamma}){n_R}{\sqrt{\gamma}} - g_1(\gamma + \sqrt{\gamma}) - \frac{g_2(\gamma+\sqrt{\gamma})}{m} - O\left(\sqrt{\gamma}\log \gamma\right).\nonumber 
\end{align}
We can see that as $\gamma\rightarrow 0$ this approaches the lower bound obtained in the exact case, so this bound is equal to that one plus terms that go to zero as $\gamma\rightarrow 0$. 
To obtain the error terms, we expand in a series around $\gamma=0$, obtaining
\begin{align}
    S(R)_\psi \geq \lambda -2\sqrt{2}n_{Q}\gamma^{1/8}- \tilde{O}(n_R\gamma^{1/4})- \tilde{O}(n_{Q}\gamma^{1/4}) \,\,\,\text{as}\,\,\, \gamma\rightarrow 0
\end{align}
where the $\tilde{O}$ notation hides logarithmic factors.
For $\gamma$ small enough, concretely $\gamma$ such that $\gamma$ is much smaller than $\min\{1/n_R^4,1/n_{Q}^{4}\}$, we obtain the lower bound
\begin{align}
    S(R)_\psi \geq \lambda - 2\sqrt{2}n_{Q}\gamma^{1/8}
\end{align}
 as claimed. 
\end{proof}

\subsubsection{Entanglement lower bound}\label{sec:entanglementbound}

We've given lower bounds on the entropy of one end of the resource system for any resource state that allows a unitary $U_{AB}$ with controllable entanglement to be implemented as a NLQC. 
In the case where the resource is pure, this provides a lower bound on entanglement in the resource. 
In the mixed case however, this is no longer true. 
Thus some further work is needed to give an entanglement lower bound in the mixed case. 

\vspace{0.2cm}
\noindent \textbf{Theorem \ref{thm:mainlowerbound}} \emph{Suppose that an NLQC protocol using resource state $\Psi_{LR}$ implements a unitary $U_{AB}$ to within diamond norm distance $\gamma$, and where $U_{AB}$ has $\lambda$-controllable entanglement.
Then for small enough $\gamma$, we have that
\begin{align}
    E_f(L:R)_\Psi \geq \lambda -P(d_{AB}\sqrt{2\gamma})
\end{align}
where $P(x)=2\sqrt{2}n_{Q}x^{1/8}$. }

\begin{proof}
Consider any decomposition of $\Psi$ into a convex sum of pure states,
\begin{align}
    \Psi_{LR}=\sum_x p_x \ketbra{\psi_x}{\psi_x}_{LR}.
\end{align}
Then, the protocol implemented can be seen as a probabilistic mixture of protocols which take in the (pure) resource state $\ket{\psi_x}_{LR}$. 
Let $\mathcal{P}$ denote the channel implemented by the protocol given resource $\Psi_{LR}$, and $\mathcal{P}_x$ the channel implemented when given resource $\psi_x$, so that $\mathcal{P}=\sum_xp_x\mathcal{P}_x$.
Let $\mathcal{U}(\cdot)=U(\cdot)U^\dagger$ be the channel formed by acting with unitary $U$. 
Then using $\Vert \mathcal{P}-\mathcal{U}\Vert_\diamond\leq \gamma$ and lemma \ref{lemma:Uextreme}, we have
\begin{align}
    \sum_xp_x \left\Vert \mathcal{P}_x-\mathcal{U}\right\Vert_\diamond\leq d_{AB}\sqrt{2\gamma}.
\end{align}
Define $\left\Vert \mathcal{P}_x-\mathcal{U}\right\Vert_\diamond=\gamma_x$, so that the above reads $\sum_x p_x\gamma_x\leq d_{AB}\sqrt{2\gamma}$. 
Then our lower bound given by theorem \ref{thm:entropylowerboundapproximate} applied to each $\mathcal{P}_x$ separately tells us that
\begin{align}
    S(R)_{\psi_x} \geq \lambda -P(\gamma_x).
\end{align}
But now consider the entanglement of formation for $\Psi_{LR}$, 
\begin{align}
    E_f(L:R)_\Psi &= \min_{\{p_i,\ket{\psi_i}\}} \sum_i p_i S(R)_{\psi_i}, \nonumber \\
    &\geq\lambda-\min_{\{p_i,\ket{\psi_i}\}} \sum_i p_i P(\gamma_i) \nonumber \\
    &\geq \lambda-\min_{\{p_i,\ket{\psi_i}\}} P\left(\sum_ip_i\gamma_i\right) \nonumber \\
    &\geq \lambda-P(d_{AB}\sqrt{2\gamma}).
\end{align}
where the last inequality used concavity of $P(\cdot)$. 
\end{proof}

A convenient property of our lower bound technique is that parallel repetition holds, as follows from the following remark. 
\begin{remark}
    \textbf{(Parallel repetition)}
    Suppose that a unitary $U$ has $\lambda$-controllable entanglement. 
    Then $U^{\otimes n}$ has $n\lambda $-controllable entanglement.
\end{remark} 
This is straightforward to verify from the definition of the controllable entanglement. 

\subsection{Evaluating the lower bound in simple cases}

Next, we begin exploring the value of the controllable entanglement for some simple gates. 
We do so numerically. 
In particular, we compute $\lambda$ for various explicit choices of gate. 
While for $\CNOT$ and gates related by local unitaries it is straightforward to guess the optimal choice of $\ket{\phi^1}$ and a suitable choice of $\ket{\phi^2}$, for general gates we perform a numerical optimization procedure to find these values. 
A program is available online\footnote{\url{https://drive.google.com/file/d/1fVHwtfde0H9wJ3QOVrTUSHkr2dhpink1/view?usp=drive_link}} which the reader can download and run to study lower bounds on any chosen two qubit gate. 

We describe our numerical approach briefly here and report our results in table \ref{table:lowerbounds}. 
To obtain our lower bound, we need to compute $\lambda$ and find an input for which $\rho^2_{QA}$ is separable, if it exists. 
Recall that $\lambda$ is equal to the entanglement of formation in system $Q\!:\!A$ with input $B$ chosen to be $\ket{\phi^1}$. 
The entanglement of formation has a closed form expression for two qubit systems \cite{hill1997entanglement} which we make use of.
Because $E_f$ is convex in its input state, we are assured that the maximizing choice of input on $B$ is a pure state. 

To find a suitable value of $\phi^2_B$ is somewhat more challenging.
To do this we define
\begin{align}
    \lambda'=\min_{\phi_2^B}E_R(Q:A)_{\phi_B^2}
\end{align}
where
\begin{align}
    E_R(Q:A)_\rho = \min_{\sigma\in SEP} D(\rho_{QA}\Vert\sigma_{QA}).
\end{align}
This is a faithful measure of entanglement, so if when we find that $\lambda'=0$, we know $\rho^2_{QA}$ is separable as needed. 
To avoid issues with numerical errors, in practice we check by hand candidate values of $\phi^2_B$ are in fact separable. 
We use the Ket package in Julia\footnote{\url{https://juliapackages.com/p/ket}} to compute $E_R$ numerically for our states of interest. 
We restrict to two qubit states, where this optimization is easy.
We then search over the space of single qubit input states to find a locally minimal choice of state $\ket{\phi^1}$, which again is pure because of convexity of $E_R$. 

We find non-trivial lower bounds for a number of common gates. 
In addition to the CNOT gate, we find lower bounds for the Berkeley B gate and for the $XX$ interaction with a rotation by (for instance) $\pi/4$.
See appendix \ref{sec:gates} for matrix expressions for these gates.
Notice $RXX(\pi/2)=\exp(-i\frac{\pi}{4}X\otimes X)$ and $\CNOT$ obtain the same numerical lower bound. 
This leaves open the possibility that these two gates are equivalent up to local operations, but we can confirm this is not the case by observing that $\CNOT\ket{\Psi^+}=\ket{0}\ket{+}$ is product, while $RXX\ket{\Psi^+}\propto \ket{\Psi^+}$ is entangled. 

Finally, we comment that our lower bound fails to be non-trivial (to give a positive lower bound) for many choices of gate. 
This occurs, for example, for the SWAP gate where we find a lower bound of 0. 
In this and related examples the failure of our lower bound is expected, since SWAP can be implemented without entanglement and using only quantum communication.
On the other hand, for many gates including Controlled-$T$, Controlled-$S$, and random two qubit unitaries our technique does not give a lower bound, but from the controllable correlation technique we can see that these gates do require entanglement, so the technique fails in these cases. 

\section{Discussion}\label{sec:discussion}

In this work we introduced the controllable correlation and controllable entanglement, and gave lower bounds on NLQC in terms of these properties. 
We showed that many simple two qubit gates can be lower bounded by these techniques. 

Computing the controllable correlation numerically allows us to explore the entanglement requirements of many gates. 
Doing so, we have made some observations that so far do not have analytical explanations. 
For instance, we computed the controllable correlation lower bound for 100,000 samples from the Haar distribution on two qubit unitaries and found a positive lower bound for every sample.
For gates that are known to not require entanglement, like the SWAP gate or product unitaries, the bound of course returns 0 correctly, but the set of such gates seems to have small (perhaps zero) measure. 
It would be interesting to extend our numerics to larger numbers of qubits, and, if the observation persists, to look for an analytical proof that unitaries with zero entanglement cost are rare. 

Another natural question is whether, when a unitary has non-zero entanglement cost, this always leads to a non-zero controllable correlation for some choice of state $P_{QA}$.
If so, this would make the controllable correlation a ``faithful'' measure of entanglement cost, analogous to a faithful measure of entanglement in quantum states. 
A related question is to begin with the assumption that the controllable correlation is $0$ for all states $P_{QA}$, and try to extract structural properties of the unitary. 
For instance, does this imply the unitary is locally equivalent to either SWAP or identity? 
Or are there other unitaries with this property?

We showed that the CNOT gate requires a resource system with $E_f$ of 1, the same entanglement of formation as one EPR pair. 
Another interesting problem is to ask if whenever the CNOT gate can be implemented exactly, the resource state is of the form $\Psi^+_{LR}\otimes \Psi_{L'R'}$ up to local unitaries on $LL'$ and $RR'$. 
This would be analogous to the self-testing property of certain non-local games.

Our techniques so far assume that the NLQC implements a unitary operation. 
This requirement is necessary, and neither the controllable correlation nor the controllable entanglement provide lower bounds without placing this requirement. 
This is clear, since if it applied to quantum channels we could obtain a lower bound on entanglement cost of $\mathcal{N}_{AB\rightarrow A}(\cdot)=\tr_B(U_{AB}\cdot U^\dagger_{AB})$, which has the same values of $\lambda_1$ and $\lambda_2$ as $U_{AB}$ does. 
However this channel can always be implemented with zero entanglement (by sending both inputs to the left), so this is a contradiction. 
An important open problem then is to find new quantities, or modifications of the controllable correlation / entanglement, which provide lower bounds on the entanglement cost of quantum channels. 

\vspace{0.2cm}
\noindent \textbf{Acknowledgements:} We thank Eric Culf who was involved in an earlier version of this project, and Eric Chitambar for noticing an important error in an early version of this paper. 
We thank the anonymous referees of TQC 2026 for further corrections.
Research at the Perimeter Institute is supported by the Government of Canada through the Department of Innovation, Science and Industry Canada and by the Province of Ontario through the Ministry of Colleges and Universities. RC is partially supported by Canada's NSERC.

\appendix

\section{List of quantum gates}\label{sec:gates}

In this appendix we give matrix expressions for the two qubit unitaries appearing in table~\ref{table:lowerbounds}. 

\begin{align}
\CNOT &= \begin{pmatrix}
        1 & 0 & 0 & 0\\
        0 & 1 & 0 & 0\\
        0 & 0 & 0 & 1\\
        0 & 0 & 1 & 0
    \end{pmatrix} 
\end{align}
\begin{align}
\text{DCNOT} &=\begin{pmatrix}
        1 & 0 & 0 & 0\\
        0 & 0 & 1 & 0\\
        0 & 0 & 0 & 1\\
        0 & 1 & 0 & 0
    \end{pmatrix}
\end{align}
\begin{align}
\text{B} &= \begin{pmatrix}
        \cos(\pi/8) & 0 & 0 & i\sin(\pi/8)\\
        0 & \cos(3\pi/8) & i \sin(3\pi/8) & 0 \\
        0 & i\sin(3\pi/8) & \cos(3\pi/8) & 0 \\
        i\sin(\pi/8) & 0 & 0 & \cos(\pi/8)
    \end{pmatrix}
\end{align}
\begin{align}
\text{RXX}(\pi/2) &=\exp\left(-i\frac{\pi}{4}X\otimes X\right) = \begin{pmatrix}
        \cos(\pi/8) & 0 & 0 & -i\sin(\pi/8)\\
        0 & \cos(\pi/8) & -i\sin(\pi/8) & 0 \\
        0 & -i\sin(\pi/8) & \cos(\pi/8) & 0 \\
        -i\sin(\pi/8) & 0 & 0 & \cos(\pi/8).
    \end{pmatrix}
\end{align}
\begin{align}
i\SWAP&=\begin{pmatrix}
        1 & 0 & 0 & 0\\
        0 & 0 & i & 0\\
        0 & i & 0 & 0\\
        0 & 0 & 0 & 1
    \end{pmatrix}
\end{align}
\begin{align}
\sqrt{\SWAP}&=\begin{pmatrix}
        1 & 0 & 0 & 0\\
        0 & \frac{1}{2}(1+i) & \frac{1}{2}(1-i) & 0\\
        0 & \frac{1}{2}(1-i) & \frac{1}{2}(1+i) & 0\\
        0 & 0 & 0 & 1
    \end{pmatrix}
\end{align}
\begin{align}
\text{Sycamore}&=\begin{pmatrix}
        1 & 0 & 0 & 0\\
        0 & 0 &-i & 0\\
        0 & -i & 0 & 0\\
        0 & 0 & 0 & e^{-i\pi/6}
    \end{pmatrix}
\end{align}
\begin{align}
\text{Magic}&=\frac{1}{\sqrt{2}} \begin{pmatrix}
        1 & i & 0 & 0\\
        0 & 0 & i & 1\\
        0 & 0 & i & -1\\
        1 & -i & 0 & 0
    \end{pmatrix}
\end{align}
\begin{align}
\text{Dagwood Bumstead}&= \begin{pmatrix}
        1 & 0 & 0 & 0\\
        0 & \cos(3\pi/8) & -i\sin(3\pi/8) & 0 \\
        0 & -i\sin(3\pi/8) & \cos(3\pi/8) & 0 \\
        0 & 0 & 0 & 1
    \end{pmatrix}
\end{align}
\begin{align}
\text{CS}&= \begin{pmatrix}
        1 & 0 & 0 & 0\\
        0 & 1 & 0 & 0 \\
        0 & 0 & 1 & 0 \\
        0 & 0 & 0 & i
    \end{pmatrix}
\end{align}
\begin{align}
\text{CT}&= \begin{pmatrix}
        1 & 0 & 0 & 0\\
        0 & 1 & 0 & 0 \\
        0 & 0 & 1 & 0 \\
        0 & 0 & 0 & e^{i\pi/4}
    \end{pmatrix}
\end{align}
\begin{align}
\text{Echoed cross resonance}&=\frac{1}{\sqrt{2}} \begin{pmatrix}
        0 & 0 & 1 & i\\
        0 & 0 & i & 1 \\
        1 & -i & 0 & 0 \\
        -i & 1 & 0 & 0
    \end{pmatrix}
\end{align}
\begin{align}
    CSX = \begin{pmatrix}
        1 & 0 & 0 & 0\\
        0 & 1 & 0 & 0 \\
        0 & 0 & e^{i\pi/4}/\sqrt{2} & e^{-i\pi/4}/\sqrt{2} \\
        0 & 0 & e^{-i\pi/4}/\sqrt{2} & e^{i\pi/4}/\sqrt{2}
    \end{pmatrix}
\end{align}

\bibliographystyle{unsrtnat}
\bibliography{biblio}

\end{document}